\documentclass[preprintnumbers,floatfix,superscriptaddress,nofootinbib,prx,twocolumn,showpacs]{revtex4-1}
\usepackage{amsmath,amsthm,amssymb}
\usepackage{graphicx}
\usepackage[english]{babel} 
\usepackage[T1]{fontenc} 
\usepackage[utf8]{inputenc} 
\usepackage{graphicx}
\usepackage{cleveref}

\usepackage{fancyhdr}
\usepackage{xcolor}
\usepackage{braket}
\usepackage{amsmath}
\usepackage{amssymb}
\usepackage{amsthm}
\usepackage{enumitem}
\usepackage{multirow}
\usepackage{lipsum}
\usepackage{changepage}
\usepackage{physics}
\usepackage{framed}
\usepackage{wasysym}

\usepackage{thmtools}
\usepackage{thm-restate}

\usepackage{floatrow}
\newfloatcommand{capbtabbox}{table}[][\FBwidth]

\usepackage{dsfont}
\usepackage{pifont}

\newcommand{\change}[1]{\textcolor{black}{#1}}
\newcommand{\steph}[1]{\textcolor{black}{#1}}

\renewcommand{\ket}[1]{|{#1}\rangle}
\newcommand{\tn}{\textnormal}
\newcommand{\dbar}[1]{\bar{\bar{#1}}}

\definecolor{cinnabar}{rgb}{0.89, 0.26, 0.2}
\definecolor{forestgreen}{rgb}{0.13, 0.55, 0.13}
\newcommand{\ctr}{\hat{\mathcal{C}}}

\theoremstyle{definition}

\newtheorem{definition}{Definition}

\newtheorem{property}{Property}
\newtheorem{protocol}{Protocol}

\newtheorem{outline}{Outline}
\theoremstyle{plain}
\newtheorem{lemma}{Lemma}
\declaretheorem[name=Theorem]{thm}
\newtheorem{corollary}{Corollary}

\usepackage{array}
\newcolumntype{C}[1]{>{\centering\let\newline\\\arraybackslash\hspace{0pt}}m{#1}}

\begin{document}

\title{Secure multi-party quantum computation with few qubits}
\vspace{2em}

\author{Victoria Lipinska}\email{v.lipinska@tudelft.nl}
\affiliation{QuTech, Delft University of Technology, Lorentzweg 1, 2628 CJ Delft, The Netherlands}
\affiliation{Kavli Institute of Nanoscience, Delft University of Technology, Lorentzweg 1, 2628 CJ Delft, The Netherlands}
\author{J\'{e}r\'{e}my Ribeiro}
\affiliation{QuTech, Delft University of Technology, Lorentzweg 1, 2628 CJ Delft, The Netherlands}
\affiliation{Kavli Institute of Nanoscience, Delft University of Technology, Lorentzweg 1, 2628 CJ Delft, The Netherlands}
\author{Stephanie Wehner}
\affiliation{QuTech, Delft University of Technology, Lorentzweg 1, 2628 CJ Delft, The Netherlands}
\affiliation{Kavli Institute of Nanoscience, Delft University of Technology, Lorentzweg 1, 2628 CJ Delft, The Netherlands}

\begin{abstract}
We consider the task of secure multi-party distributed  quantum computation on a quantum network.  We propose a protocol based on quantum error correction which reduces the number of necessary qubits. That is, each of the $n$ nodes in our protocol requires an operational workspace of $n^2 + 4n$ qubits, as opposed to previously shown $\Omega\big((n^3+n^2s^2)\log n\big)$ qubits, where $s$ is a security parameter. Additionally, we reduce the communication complexity by a factor of $\mathcal{O}(n^3\log(n))$ qubits per node, as compared to existing protocols. 
To achieve universal computation, we develop a distributed procedure for verifying magic states, which allows us to apply distributed gate teleportation and which may be of independent interest. We showcase our protocol on a small example for a 7-node network. \\
\indent NOTE FROM AUTHORS: This is a version with an erratum appended, see below. The implemented changes increase the operational qubit workspace per node from $n^2 + 4n$ to $n^2+\Theta(s) n$, where $s$ is the security parameter. The increase is linear in $n$, which means that the main result of our paper remains intact: number of qubits per node necessary to implement the multiparty quantum computation {is still smaller than} the previously existing protocols. {Moreover, the security proof in our manuscript does not change.}
\end{abstract}

\maketitle

\section{Introduction}

Secure multi-party computation is a task which allows $n$ nodes of a network to jointly compute a function on their inputs \cite{Yao1982}. The inputs are private, meaning that they are only known to the nodes who supplied them. What is more, the only information that can be inferred about the private inputs is whatever can be inferred from the outputs of the computation and the computation itself. Multi-party computation allows for distributed evaluation of any function, and hence it is a powerful cryptographic primitive with many practical (e.g.  clearing a commodity derivative market) and theoretical (e.g. zero knowledge proofs) applications \cite{Cramer2015}.

In the domain of quantum computation the problem of multi-party quantum computation (MPQC) on quantum data was first introduced by \cite{Crepeau2002}. It can be defined as follows: each node $i = 1,\dots,n$ gets one, possibly unknown, input quantum state $\rho_i$. The nodes jointly perform an $n$-input arbitrary quantum circuit $\mathfrak{R}$ on their inputs $\rho_1,\dots, \rho_n$. The output of the circuit is divided into $n$ parts and each node $i$ gets $i$-th part of the output state, see Figure \ref{fig:schematicR}. In MPQC there can be nodes who do not follow the protocol (cheaters). We then require that an MPQC protocol satisfies the following informal requirements:
\begin{itemize}[topsep=0pt,itemsep=-1ex,partopsep=1ex,parsep=1ex]
\item (Correctness) If there are no cheaters, then the protocol implements the intended circuit $\mathfrak{R}$ on the inputs of the nodes.
\item (Soundness) Cheaters cannot affect the outcome of the computation of the other nodes, beyond their ability to choose their own inputs.
\item (Privacy) Cheaters do not learn anything about private inputs and outputs of the other nodes.
\end{itemize}
Throughout this paper we will consider that an input $\rho_i$ of each node is a single-qubit state.

\begin{figure}[b]
\ffigbox{%
 \includegraphics[width=0.6\textwidth]{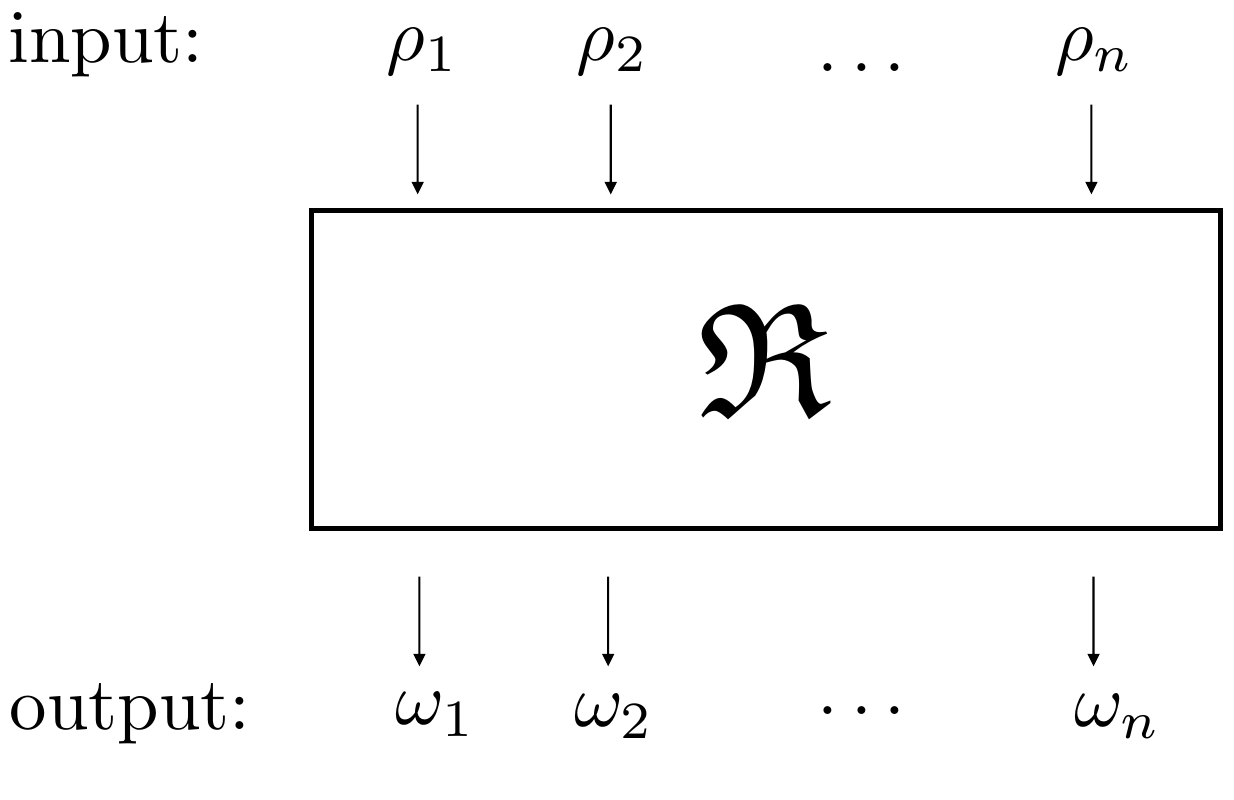}%
}{
\caption{Each of the nodes $1,\dots,n$ provides a single-qubit input $\rho_1,\dots,\rho_n$. The goal of the multi-party quantum computing (MPQC) protocol is to implement circuit $\mathfrak{R}$ such that each node gets an output $\omega_1,\dots,\omega_n$ without gaining any knowledge of the other inputs or outputs beyond their ability to choose their own inputs. Note that the inputs (and outputs) can be entangled.}\label{fig:schematicR}
}
\end{figure}

The approach taken by the original work of \cite{Crepeau2002} is based on a subroutine called verifiable quantum secret sharing and is a generalization of a classical multi-party computation \cite{Chaum1988_unconditionally}. The security achieved by the protocol is information theoretical, meaning that the cheaters are not constrained by computational assumptions. However, the number of cheaters has to be strictly smaller than $\frac{n}{6}$. This bound was later lifted to $\frac{n}{2}$ by \cite{Crepeau2005}, who used authentication schemes and approximate error correction. However this solution requires significantly more qubits to be realized.
At the same time, there exist parallel approaches tolerating a cheating majority and whose security relies on computational assumptions, for example \cite{Dupuis2012} for the case of $n=2$ or its recent generalization to $n>2$ \cite{Dulek2019}. 
Note that a protocol tolerating more than $\frac{n}{2}$ cheaters is not possible without additional computational assumptions, since that would imply the existence of unconditionally secure bit commitment \cite{Mayers1997,Lo1998}.

\begin{table*}[t]

\caption{Summary of qubit savings presented in this paper, $s$ denotes the security parameter of the protocol, $\#\tn{ancillas}$ denotes the number of ancillas in circuit $\mathfrak{R}$, and $\# T$ denotes the number of $T$ gates. The size of the workspace in our protocol does not depend on the security parameter, because of the sequential execution of the verification phase, see Section \ref{subsec:VHSS}. Note that here we do not list the work of \cite{Dulek2019}, since their protocol does not use techniques based on error correction and achieves computational security guarantees. }\label{tab:qubit_savings}
\begin{tabular}{|c|c|c|}
\hline 
&Our protocol & Crepeau et al. \cite{Crepeau2002} \\ 
\hline 
size of the input in qubits per node & $1$ & $\Omega(\log n)$ \\ \hline
size of an individual share during the computation & \multirow{2}{*}{$1$}  & \multirow{2}{*}{$\Omega(\log n)$} \\ 
in qubits per node & & \\ \hline
$\#$ qubits in workspace per node &$n^2 + 4n$& $\Omega\big((n^3+n^2s^2)\log n\big)$ \\ \hline
$\#$ qubits sent per node & $\mathcal{O}\big((n + \#\tn{ancillas} + \# T)ns^2 \big)$ & $\mathcal{O} \big((n^2 + \#\tn{ancillas})n^3s^2\log(n)\big)$\\
\hline 
\end{tabular} 
\end{table*}

In this work we are interested in the former approach to MPQC, namely the one based on verifiable quantum secret sharing of \cite{Crepeau2002}. Our objective is to perform MPQC on a quantum network with $n$ nodes using as few qubits as possible. The approach we take is based on \cite{Crepeau2002} and extensively relies on techniques from fault-tolerant quantum error correction. It can be intuitively understood as follows. Nodes use a chosen quantum error correcting code and create a global logical state $\bar{\Psi}$ by encoding each of the single-qubit input states. Each node holds a part of this logical state, we call such a part a \textit{share}. They verify the encoding of each  state using verifiable secret sharing protocol and perform local operations to evaluate a logical version of the circuit $\mathfrak{R}$, and then locally reconstruct their outputs. 

To be able to apply any circuit $\mathfrak{R}$ this way, we need two properties. First, $\mathfrak{R}$ needs to be composed of gates which form a universal set, i.e. any circuit can be decomposed into gates from that set. Second, if the nodes apply only local operations $\Lambda$ from the universal set, it should yield a meaningful logical operation $\bar{\Lambda}$ for the global state $\bar{\Psi}$. This property is called \emph{transversality}. However, for any error correcting code, it is impossible to perform universal quantum computation using only transversal gates \cite{Eastin2009}. For this reason, it is common to extend a transversal set of gates (for example Clifford gates) with a non-transversal gate (for example $T$ gate or the Toffoli gate). Note that there exist methods to realize single non-transversal gates in a distributed way, for example by using ancilla states \cite{Gottesman1999} or locally modifying the error correcting code \citep{Aharonov1997}.

In particular, \cite{Crepeau2002} considers quantum polynomial codes and a universal set of gates with the Toffoli gate \citep{Aharonov1997}. This solution is very expensive in qubits. Firstly, the polynomial codes require local shares whose dimension scales with the number of nodes, and therefore require $\Omega(\log n)$ qubits per share. Moreover, the nodes need to perform a distributed encoding of the shares in order to apply the Toffoli gate. This means that each  input state must be encoded three times using the polynomial code. Performing the three-level encoding serves one more purpose, namely, it localizes all of the errors  in the encoding to the positions of the cheaters. As a result, the cheaters cannot force the protocol to abort, since any error they introduce will always be corrected by the underlying polynomial code. All in all, each node needs an operational workspace of $\Omega\big((n^3+n^2s^2)\log n\big)$ qubits, where $s$ is the security parameter of the protocol, see Table \ref{tab:qubit_savings}. We remark that in schemes based on exact error correcting codes, the number of cheaters $t$ is intrinsically constrained by the distance $d$ of the underlying code as $t\leq \left\lfloor \frac{d-1}{2} \right\rfloor$, which in principle can reach $\frac{n}{4}$ \cite{Rains1997,Grassl2004}. However, the technique for applying the Toffoli gate in \cite{Crepeau2002} puts a constraint on the number of cheaters to~$\frac{n}{6}$.

Since near-term quantum networks will be able to support only a small number of qubits, it would be preferable to implement an MPQC protocol with as few qubits as possible. So far, reducing quantum resources has received a lot of attention in the domain of non-distributed quantum computation and simulation, see for example \cite{Bravyi2016,Steudtner2018,Moll2016,Bravyi2017,Peng2019}. 
Recently, in \cite{Lipinska2019} we considered a problem of reducing quantum resources 
for a distributed protocol, namely verifiable secret sharing of a quantum state. 
Here we address a similar issue of whether distributed multi-party quantum computation can be performed on a quantum network with less quantum resources. We answer this question positively by proposing a scheme for universal distributed computation which uses fewer qubits as compared to the existing approach of \cite{Crepeau2002} outlined above.

This paper is organized as follows. In Section \ref{sec:results} we summarize our contributions, where in \ref{sec:resources} we discuss the implications of our protocol on resource reduction and in \ref{sec:example} we give an explicit example of the protocol on a 7-node network. In Section \ref{sec:methods} we zoom in on the technical aspects of our work. There we present the protocol in detail and provide formal security statements. We leave out technical proofs for Appendix \ref{app:proofs}.

\section{Results}\label{sec:results}

We propose a protocol for secure multi-party quantum computation where each node holds single-qubit shares. Our approach is based on quantum error correcting codes, similar to the idea of \cite{Crepeau2002,Smith2001,Crepeau2005}. Since our interest lies in reducing the quantum resources necessary to realize the protocol, we abandon the original idea of three-level encoding at the cost of allowing the protocol to abort if the initial encoding of the shares is incorrect. Thanks to this, we are able to execute the protocol with less qubits in the workspace per node and lower communication complexity, see Table \ref{tab:qubit_savings}. Moreover, we develop a procedure for a distributed verification of any logical state which is stabilized by a Clifford gate. This allows us to perform distributed gate teleportation and implement a universal set of gates without creating three levels of encoding. What is more, we follow the approach outlined in \cite{Lipinska2019} which allows for a sequential execution of the verification of the inputs. This solution reduces the operational workspace to $n^2 + 4n$ qubits per node.  We show that our protocol is secure in the presence of active non-adaptive cheaters (see Adversary model), where the number of cheaters is constrained by the distance $d$ of the underlying error correcting code, i.e. $t\leq \left\lfloor \frac{d-1}{2} \right\rfloor$. 
Finally, we showcase our protocol on a small example for 7 nodes using Steane's 7-qubit code \cite{Steane1996}.

The key to our results is using error correcting codes which encode a single qubit into $n$ single qubits. Furthermore, we allow the MPQC protocol to abort if there are too many errors detected during the execution of the protocol. It is, however, possible to execute our protocol without the abort event, see Section \ref{sec:outlook}. This solution requires many more rounds of communication and we do not consider this approach explicitly. We develop a distributed verification procedure for magic states which allows us to implement gate teleportation in a distributed way.  Lastly, we verify the input of each node using ancillas distributed in a sequential way. We elaborate on these techniques in the next section, Section \ref{sec:resources}.

\begin{framed}
\begin{outline}[Multi-party quantum computation]~\\
Input: single-qubit state $\rho_i$ from each node, CSS code $\ctr$ with transversal Cliffords, circuit $\mathfrak{R}$.
\begin{enumerate}
\item \textit{Sharing and verification}\\
Each node $i=1,\dots,n$ encodes her input $\rho_i$ using code $\ctr$ into an $n$-qubit logical state, and sends one qubit (i.e. one single-qubit share) of the logical state to every other node, while keeping one for herself.
The nodes jointly verify the encoding done by node $i$ using verifiable quantum secret sharing protocol (see Protocol~\ref{prot:vhss}).
\item \textit{Computation}
\begin{itemize}[topsep=0pt,itemsep=-1ex,partopsep=1ex,parsep=1ex]
\item For every Clifford gate in circuit $\mathfrak{R}$:\\
The nodes apply transversal Clifford gates locally to qubits specified by the circuit $\mathfrak{R}$.\\
\item For every $T$ gate in circuit $\mathfrak{R}$ applied to qubit $i$:\\
Node $i$ prepares the magic state $\ket{m} = \frac{1}{\sqrt{2}} (\ket{0} + e^{i\frac{\pi}{4}}\ket{1})$. The nodes verify it using Verification of Clifford-Stabilized States protocol, see Protocol \ref{prot:verification_stabilizer}. If the verification is successful, the nodes perform Distributed Gate Teleportation, see Protocol \ref{prot:gate_teleport}. 
\end{itemize}
Every $\ket{0}$ ancilla state required for circuit $\mathfrak{R}$, which is prepared by node $i$, is jointly verified by the nodes using verifiable quantum secret sharing, Protocol \ref{prot:vhss}.

If the verification of any step fails the nodes substitute their shares for $\ket{0}$ and abort the protocol at the end of the computation.

\item \textit{Reconstruction}\\
Each node $i$ collects all shares of her part of the output. She corrects errors using code $\ctr$ and reconstructs her output.

\end{enumerate}

\end{outline}
\end{framed}

\vspace{1em}
\textbf{Network model.} We consider a quantum network with $n$ nodes. Each node can locally process $\mathcal{O}(n^2)$ qubits, and can perfectly process and store classical information. Each pair of nodes is connected via private and authenticated classical \cite{Canetti2004} and quantum \cite{Barnum2002} channels. Additionally, we assume that the nodes have access to an authenticated classical broadcast channel \cite{Canetti1999} and a public source of randomness. Note that a source of randomness can be created, for example, by running a classical multi-party computation protocol \cite{Rabin1989}.

\textbf{Adversary model.} We say that $t$ out of $n$ nodes are \emph{active} cheaters during the protocol. This means that they can act maliciously throughout the entire execution of the multi-party computation and perform arbitrary joint quantum operations on their shares, possibly with quantum side information. Therefore, the security of our protocol does not rely on computational assumptions. 
We assume that the active cheaters are \emph{non-adaptive}, meaning that they are determined prior to the beginning of the protocol and stay fixed throughout its execution. On the other hand, the nodes which follow the protocol exactly are \emph{honest}.  A protocol \emph{tolerates} the presence of $t$ active cheaters if they cannot influence the output of the protocol beyond choosing their own inputs.

\subsection{Techniques}\label{sec:resources}

Thanks to using single-qubit error correcting codes, distributed verification of magic states, the possibility to abort the protocol and sequential verification of the inputs, our MPQC protocol lowers the number of qubits that each node needs to control and send. Here we discuss in detail all the reductions made by our protocol. Then, we give an explicit example of a protocol based on the 7-qubit Steane's code.

\begin{itemize}
\item \textbf{Single-qubit CSS codes.} We consider a class of Calderbank-Shor-Steane (CSS) error correcting codes \cite{Steane1996,Calderbank1996}, which encode a single logical qubit into $n$ physical qubits, and for which applying Clifford gates is transversal, see Section \ref{subsec:CSS} for details. In particular, this means that each input state and each encoded ancilla is encoded and distributed using single-qubit shares. For comparison, the protocol of \cite{Crepeau2002} uses a class of polynomial codes, called Reed-Solomon codes \cite{Aharonov1997}, where the size of individual share grows with the number of nodes $n$ in the network as $\Omega(\log n)$ qubits.

\item \textbf{MPQC with abort.} We introduce an ``abort'' event in the MPQC protocol. That is, the protocol aborts if there are more than $t$ errors introduced by the cheaters,
accumulated over all inputs. This condition is necessary, since applying a transversal gate between different logical inputs can still propagate errors between them. As a result, we are able to perform the MPQC protocol on the two-level encoding created by the verifiable quantum secret sharing (VQSS) subroutine, see Section \ref{sec:subroutines}. This allows us to achieve a lower communication complexity -- in our protocol each node sends $\mathcal{O}\big((n + \#\tn{ancillas} + \# T)ns^2 \big)$ qubits, as opposed to $\mathcal{O} \big((n^2 + \#\tn{ancillas})n^3s^2\log(n)\big)$ qubits in \cite{Crepeau2002}, where $s$ denotes the security parameter of the protocol, $\#\tn{ancillas}$ denotes the number of ancillas in circuit $\mathfrak{R}$, and $\# T$ denotes the number of $T$ gates. Note that in our protocol we can avoid the abort event by creating the third level encoding, following the idea of \cite{Crepeau2002}. This approach confines the errors of all inputs only to the positions of $t$ cheaters, see Section \ref{sec:outlook} for discussion. However, this solution significantly increases quantum communication complexity. Since our objective is to reduce the number of qubits, we do not consider this approach here. 

\item \textbf{Verification of Clifford-stabilized states.} We develop a distributed method for verifying states stabilized by the Clifford gates, which in particular can be applied to verify magic states. This solution allows us to perform distributed gate teleportation and apply the $T$ gate in a distributed way. Recall that for our MPQC protocol we choose CSS codes with transversal Clifford gates. This, together with distributed gate teleportation and transversal measurements, provides a way to apply a universal set of gates in a distributed way. Thanks to using magic state ancillas, we can perform the computation on a two-level encoding created during the verification phase (see Protocol \ref{prot:mpqc}). This means that each node controls $n^2$ single-qubit shares of all inputs. In contrast, in the approach of \cite{Crepeau2002} the nodes need to apply a non-linear Toffoli gate to achieve universality of computation. This, in turn, required
a workspace of $\Omega\big((n^3+n^2s^2)\log n\big)$ qubits per node.

\item \textbf{Sequential verification.} We use the verifiable quantum secret sharing (VQSS) protocol of \cite{Crepeau2002} to verify that the encoding was carried out correctly and that at the end of the computation there will be a state to reconstruct. The verification procedure requires ancillary states. However, following the idea developed in \cite{Lipinska2019}, we perform the verification in a sequential way. That is, to verify each input we use the ancillas one by one instead of all at once as in \cite{Crepeau2002}. 
In particular, the nodes use at most $2n$ single-qubit ancillas at a time to verify the input states (or ancillas in $\mathfrak{R}$) and at most $4n$ single-qubit ancillas to apply the $T$ gate.
\end{itemize} 

All in all, this amounts to an operational workspace of at most $n^2 + 4n$ single qubit shares for our protocol. Out of those, $n^2$ shares correspond to the input states on which the distributed computation is performed. 
For comparison, the protocol of \cite{Crepeau2002} requires simultaneous control over  $\Omega\big((n^3+n^2s^2)\log n\big)$ qubits per node, where $s$ is the security parameter of the protocol. Moreover, due to introducing the possibility of aborting the protocol, our MPQC scheme lowers the communication complexity. That is, our protocol reduces the number of qubits that each nodes has to send by a factor of $\mathcal{O}(n^3\log(n))$ compared to the protocol of \cite{Crepeau2002}.

Finally, when the number of cheaters $t$ is restricted by the distance $d$ of the CSS code, i.e. when $t\leq  \left\lfloor \frac{d-1}{2} \right\rfloor$, we prove that our protocol satisfies the usual security requirements (soundness, completeness and privacy, see above). Our statements follow from the fact that any error correcting code has the ability to correct at most $ \left\lfloor \frac{d-1}{2} \right\rfloor$ arbitrary errors and therefore, any errors introduced by the cheaters can be corrected by the honest nodes. What is more, the inputs and outputs of honest nodes will be also private, since if they recover the outputs exactly, then the cheaters get no information about inputs or outputs \cite{Gottesman2000}. Our statements hold with probability exponentially close to 1 in the security parameter~$s$.

\subsection{Example for 7 nodes} \label{sec:example}

Let us consider a network of $n=7$ nodes and assume that the nodes want to perform a CNOT between inputs $\rho_1$ and $\rho_2$ of nodes ``1'' and ``2'' of the network. For the execution of this protocol we will need a workspace of 28 qubits per node. For the sake of the example, we will also assume that the inputs are pure single-qubit states, $\rho_1 = \ketbra{\psi_1}$ and $\rho_2 = \ketbra{\psi_2}$, and that the protocol does not abort. The 7-qubit Steane's code \cite{Steane1996} is the smallest example of a qubit CSS code with transversal Cliffords. This code has distance $d=3$ meaning that it can correct $\left\lfloor \frac{d-1}{2} \right\rfloor = 1$ arbitrary error. This also means that in an MPQC protocol built on the 7-qubit code, we can tolerate $t=1$ cheater. 

\textit{Sharing and verification.} Node ``1'' encodes her single-qubit pure input $\ket{\psi_1}$ into 7 physical qubits using the Steane's code encoding map $\mathcal{E}$. She sends one qubit to each of the remaining 6 nodes, while keeping one qubit to herself. Each node again encodes the received qubit using the Steane's code and shares 6 qubits of that encoding with other nodes. At this point the input state $\ket{\psi_1}$ has been encoded twice, i.e.
\begin{align}
\dbar \Psi_{1} = \mathcal{E}^{\otimes 7} \circ \mathcal{E}(\ketbra{\psi_1}),\\
\dbar \Psi_{2} = \mathcal{E}^{\otimes 7} \circ \mathcal{E}(\ketbra{\psi_2}).
 \end{align} 
Each node holds $7$ qubits in total.

The nodes run the verification procedure according to \cite{Lipinska2019}, verifying that the encoding of each node $i$ was done correctly. The encoding of each input state can be verified one at a time. In one round of verification of a single input, each node uses at most 14 local ancilla qubits. The ancillas shares are encoded twice with the 7 qubit code and distributed in the same way as the input states. The nodes randomly perform the CNOT gate between $\dbar \Psi_{1}$ and an ancilla,
 to identify errors possibly introduced by cheating nodes. These ancillas are then measured and the outcome of the measurement allows the nodes to jointly conclude whether verification of the encoding was correct, i.e. whether the distributed input states have at most $t=1$ error on the same position. If so, then the errors are correctable by the 7 qubit code, and the nodes hold a valid logical state of the code. This procedure is repeated $s^2+2s$ times in total, where $s$ is the security parameter.
 
The same sharing and verification procedure is carried out for node ``2'' and her single-qubit pure input $\ket{\psi_2}$: it is first shared and then verified. As before, the verification requires at most 14 local ancilla qubits at a time. 
After the second verification each node holds 14 verified data qubits corresponding to the logical inputs $\dbar \Psi_{1} \otimes \dbar \Psi_{2}$. 
Note that the input states are never measured.

\textit{Computation.} Each node applies the CNOT gate locally to shares coming from node ``1'' and ``2''. The CNOT gate is a Clifford gate. Therefore, since the inputs are verified to be logical states of the 7 qubit code, applying the CNOT locally is well-defined and yields a logical operation between logical inputs $\dbar \Psi_{1} \otimes \dbar \Psi_{2}$. Let us define the output of the computation $\dbar \omega$,
\begin{equation}
\dbar \omega = \dbar{CNOT}\left(\dbar \Psi_1 \otimes \dbar \Psi_2\right).
\end{equation}

\textit{Reconstruction.} Nodes ``1'' and ``2'' get all of the shares corresponding to her own outputs, i.e.
\begin{equation}
\dbar \omega_1 = \tr_{2}(\dbar \omega), \quad \dbar \omega_2 = \tr_{1}(\dbar \omega).
\end{equation}
They separately run local error correcting circuit of the 7 qubit code on $\dbar \omega_1$ and $\dbar \omega_2$, respectively. They identify errors, see Reconstruction of Protocol \ref{prot:mpqc} for details. This is necessary, since the cheater might have introduced errors during or after the computation, and right before the reconstruction. Each of the nodes ``1'' and ``2'' corrects errors and reconstructs her output $\omega_1$ and $\omega_2$, respectively. The outputs are single qubit states, and are such that
\begin{align}
\omega_1 & = \tr_2(CNOT(\ketbra{\psi_1}\otimes \ketbra{\psi_2})), \\ \omega_2 & = \tr_1(CNOT(\ketbra{\psi_1}\otimes \ketbra{\psi_2})).
\end{align}

\section{Methods}\label{sec:methods}
In this section we discuss our MPQC protocol in detail. We lay down the framework by first discussing properties of CSS codes which will be useful for the distributed computation in Section \ref{subsec:CSS}. Then we introduce a few important subroutines, namely Verifiable Secret Sharing (Section \ref{subsec:VHSS}), Distributed Gate Teleportation (Section \ref{subsec:distributed_gate_teleport}) and Verification of Clifford-Stabilized States (Section \ref{subsec:verification_magic_state}). Finally, in Section \ref{subsec:MPQC} we discuss our Multi-party Quantum Computation protocol and state its security in Section \ref{subsec:Security}.

\subsection{CSS codes}\label{subsec:CSS}
In our considerations we will focus on a class of error correcting codes called Calderbank-Shor-Steane (CSS) codes \cite{Steane1996,Calderbank1996}. A CSS code $\mathcal{C}$  is defined through two binary classical linear codes, $V$ and $W$ , satisfying $V^* \subseteq W$, where $V^*$ is the dual code of $V$. Then, $\mathcal{C} :=V \cap \mathcal{F}W$ is a set of states of $n$ qubits which yield a codeword in $V$ when measured in the standard basis, and a codeword in $W$ when measured in the Fourier basis.  A code encoding one logical qubit into $n$ physical qubits is commonly denoted with double square brackets $[[n,1,d]]$. Here $d$ is the distance of the code, which relates to the maximum number of arbitrary errors $t$ which the code can correct as $t \leq \left\lfloor \frac{d-1}{2} \right\rfloor$. 

In distributed computation each node can only apply local operations. Therefore, we want that logical operations $\bar{\Lambda}$ are implemented by applying local operations $\Lambda$ on the individual qubits held by the nodes and encoded with $\mathcal{C}$, i.e. $\bar{\Lambda} = \Lambda^{\otimes n}$. This property is called transversality.
For our construction of the MPQC protocol we choose specific CSS codes $\ctr$ with transversal operations, which satisfy:
\begin{enumerate}
\item $\ctr$ uses the same classical code to correct $X$ and $Z$ errors, i.e. $V=W$.
\item  The weight of the stabilizer generators of $\ctr$ is a multiple of 4, and the logical Pauli operators $X$ and $Z$ have weight $1\mod 4$, or $3\mod 4$. 
\end{enumerate}
Property 1 guarantees that the Hadamard gate $H$ can be applied transversally, while property 2 guarantees that the phase gate 
$P = \begin{pmatrix}
1 & 0 \\ 0 & i 
\end{pmatrix}$ can be applied transversally. Additionally, note that the CNOT gate is transversal for any CSS code. Since $H,P$ and CNOT generate the Clifford set, one can apply any Clifford gate on the code $\ctr$ transversally \cite{Gottesman1998}. Finally, any CSS code has a property that measurements can be performed qubit-wise, but the measurement outcome of every qubit must be communicated classically to obtain the result of the logical measurement.

\begin{figure*}[!t]
\ffigbox{%
 \includegraphics[width=\textwidth]{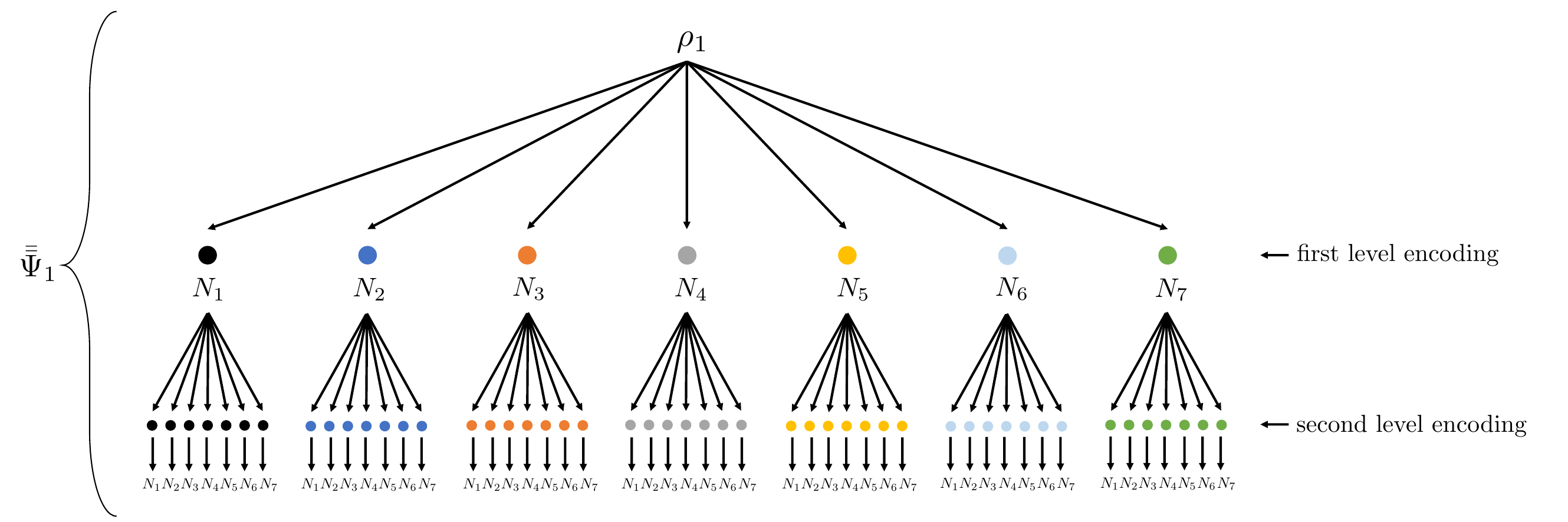}%
}{
\caption{Two-level encoding of the input qubit state $\rho_1$ of node ``1''. The double-encoded distributed state is denoted by $\dbar{\Psi}_1$. Each dot represents a single-qubit share.}\label{fig:twolevel_encoding}
}
\end{figure*}
\vspace{-1em}
\subsection{Subroutines}\label{sec:subroutines}

Here we list and describe the subroutines we will later use as building blocks in our MPQC protocol. We start with reviewing an existing construction of verifiable quantum secret sharing used for verifying inputs in MPQC. 
Next, we discuss two of our contributions -- distributed gate teleportation and verification of states stabilized by Clifford gates. These last two subroutines will be essential for implementing universal circuits in MPQC.

\subsubsection{Verifiable quantum secret sharing}\label{subsec:VHSS}

One of the first ingredients of our MPQC protocol is verifiable quantum secret sharing (VQSS) first introduced in \cite{Crepeau2002},  see Protocol \ref{prot:vhss}.  Here we use a modified version of the scheme, which we introduced in \cite{Lipinska2019} to reduce the qubit workspace required for each node. 
A VQSS scheme is a scheme
which shares a quantum state among $n$
nodes in a verifiable way using quantum shares. The scheme we use is based on a CSS code $\mathcal{C}$ with distance $d$, and tolerates at most $t\leq \left\lfloor \frac{d-1}{2} \right\rfloor$ non-adaptive active cheaters. We remark that the scheme works for any CSS code $\mathcal{C}$.

Let us describe the task in detail. In VQSS the dealer $D$ encodes her input state $\rho$ using the code $\mathcal{C}$. The encoding produces an $n$-qubit entangled state. $D$ shares this state among the nodes by sending one qubit to each node. Each node then encodes the received one-qubit share again with the same error correcting code into $n$ qubits, and sends one qubit to each of the $n$ nodes. This way each node holds $n$ single-qubit shares. We denote a double-encoded logical global state of the nodes with a double bar, $\dbar{\Psi}$. Throughout the rest of this paper we will use index $i = 1,\dots, n$ to denote the encoding performed by node $i$, and $\ell = 1,\dots, n$ to denote the share held by node $\ell$. The share held by node $\ell$ and coming from encoding performed by node $i$ will be denoted as $\dbar \Psi_{i_\ell}$. \looseness=-1

The nodes run a verification procedure to verify that $\dbar{\Psi}$ is a valid codeword of the code $\mathcal{C}$. The verification is a generalization of Steane's error correction method to the distributed setting \cite{Steane1997}. More specifically, the nodes publicly check that there are at most $t \leq \left\lfloor \frac{d-1}{2} \right\rfloor$ errors at the first level of encoding, i.e. the encoding done by the dealer. To do so, they use ancilla qubits encoded twice with the same code $\mathcal{C}$. These ancillas are measured during the verification. Since $\mathcal{C}$ is a CSS code, the measurement outcomes yield a codeword from a classical code $V$ (resp. $W$) when measured in the standard (resp. Fourier) basis. Using an error correcting procedure for the classical linear codes allows the nodes to identify shares of the first-level encoding which carry errors. The positions of these shares are collected in a public set $B$ of \emph{apparent} cheaters (indeed, there is no way to tell apart the errors introduced by the dealer and errors introduced by the cheaters on the first-level encoding). If there are at most $t$ first-level errors (i.e. $|B| \leq t$), the dealer passes the verification. Moreover, since the protocol assumes the existence of at most $t$ cheaters, there can be at most $t$ errors in each second-level encoding. Therefore, if the dealer passes the verification, at the end of the protocol there will always be a state to reconstruct, since errors at both first and second level encoding can be corrected by the code $\mathcal{C}$. Following the idea introduced in \cite{Lipinska2019}, this verification procedure can be performed by encoding and measuring one ancilla qubit at a time. There are $s^2+2s$ iterations of the verification procedure, where $s$ is the security parameter.  Additionally, similarly as in \cite{Lipinska2019}, we use CSS codes which encode a single qubit into $n$ single qubits. 
The sequential VQSS protocol requires a $3n$-qubit workspace per node to verify one single-qubit input state, see \cite{Lipinska2019} for details. Each node needs to send $\mathcal{O}(n^2s^2)$ qubits. 

\textbf{Verification of logical 0 (VQSS(0)).} 
In the following sections we will make use of a handy property of the VQSS protocol of \cite{Crepeau2002}. Namely, the protocol can verify that the state shared by the nodes is exactly the logical $\ket{\dbar{0}}$ of code $\mathcal{C}$, see \cite{Crepeau2002,Smith2001,Lipinska2019}. The verification phase is almost the same as in the VQSS  protocol of \cite{Crepeau2002}, except now the nodes check whether the classical measurement outcomes interpolate to 0 after decoding them twice with a classical decoder, see \cite{Crepeau2002,Smith2001} for details. 
We will refer to this verification procedure as VQSS(0).

\onecolumngrid
\begin{framed}
\begin{protocol}[Verifiable Quantum Secret Sharing (VQSS) \cite{Crepeau2002, Lipinska2019} - outline]\label{prot:vhss}
~\\
Input: Single-qubit state $\rho$ of dealer $D$ to share, CSS error correcting code $\mathcal{C}$.
\begin{enumerate}
\item \textbf{Sharing}\\
The dealer $D$ encodes her input $\rho$ into a logical state using code $\mathcal{C}$ and sends each qubit of the logical state to every other node, while keeping one for herself. Each node encodes the share received from $D$ again using $\mathcal{C}$ and shares among the nodes keeping one qubit for herself. Therefore, the nodes create a two-level encoding of $\rho$. At this point each node holds $n$ single-qubit shares coming from every other node.
\item \textbf{Verification}\\
Nodes verify whether $D$ is honest, i.e. that the shares held by the nodes are consistent with a codeword of $\mathcal{C}$ and at the end of the protocol a state will be reconstructed. The nodes construct a public set $B$ which records positions of nodes with inconsistent shares on the first level of encoding.\\  Each node uses at most additional $2n$ ancilla qubits for one iteration of the verification procedure. There are $s^2+2s$ iterations of verification, where $s$ is the security parameter. If $|B|\leq t$ the dealer passes the verification phase.
\end{enumerate}
\end{protocol}
\end{framed}
\twocolumngrid

\subsubsection{Distributed gate teleportation}\label{subsec:distributed_gate_teleport}

To perform universal computation, we need a universal set of gates. However, Clifford gates by themselves are not a universal set. An example of a set that \emph{is} universal, is the set generated by the Clifford gates extended with the $T = \sqrt{P}$ gate \cite{Nebe2001}, denoted Cliff$+T$\footnote{One can efficiently approximate any gate $G$  within distance $\epsilon$ using $\tn{polylog}(1/\epsilon)$ gates from set Cliff$+T$ \cite{Kitaev1997}.}. On the other hand, for any error correcting code, it is impossible to perform universal quantum computation using only transversal gates \cite{Eastin2009}. In particular, for the class of CSS codes under consideration, $\ctr$, the Clifford gates can be applied transversally (see Sec. \ref{subsec:CSS}), but the $T$ gate cannot.

To remedy this problem in the domain of quantum (non-distributed) computing, one can use a technique called gate teleportation \cite{Gottesman1999}. In particular, for the $T$ gate, the idea is to use a specially created ancilla state, measure, and apply a correction depending on the measurement outcome, see Figure \ref{fig:gate_teleport}. Importantly, this correction is done with $XP^\dagger$ and since $XP^\dagger$ is a Clifford gate, it can be applied transversally. The cost of this procedure is to create the special ancilla state, which is commonly referred to as a magic state. In the case of the $T$ gate it is $\ket{m} = \tfrac{1}{\sqrt{2}} (\ket{0} + e^{i\pi/4}\ket{1})$.

We generalize this procedure to a distributed setting, see Protocol \ref{prot:gate_teleport}. Our protocol takes two states as an input: logical $\dbar{\Psi}$ and logical $\ket{\dbar{m}}$, both encoded twice (two-level encoding) with code $\ctr$. We assume at this point that both states are verified with respect to the same dealer $D$. The verification of $\dbar{\Psi}$ can be performed with VQSS. However, verifying that $\ket{\dbar{m}}$ is \emph{exactly} the magic state is non-trivial and we introduce it in the next section. 

To apply a logical $T$ gate to $\dbar{\Psi}$ the nodes first perform a logical transversal CNOT operation on their shares, taking shares of $\ket{\dbar{m}}$ as a control and shares of $\dbar{\Psi}$ as a target. Then each node $i=1,\dots,n$ measures the target qubit in the standard basis and announces the measurement outcome. Nodes publicly check whether the measurement collapsed the target state onto a classical string corresponding to a logical $\ket{\dbar{0}}$ or a logical $\ket{\dbar{1}}$. To do so, they check whether the resulting string of measurement outcomes $\mathbf{v}_i$ interpolates to 0 or to 1 using the classical decoder twice. At the same time the nodes update the set of errors $B$. If the interpolated value is 0 then no correction is necessary. If the interpolated value is 1 then the nodes apply the correction $XP^\dagger$ transversally. 

\begin{figure}[b]
\ffigbox{%
 \includegraphics[scale=0.3]{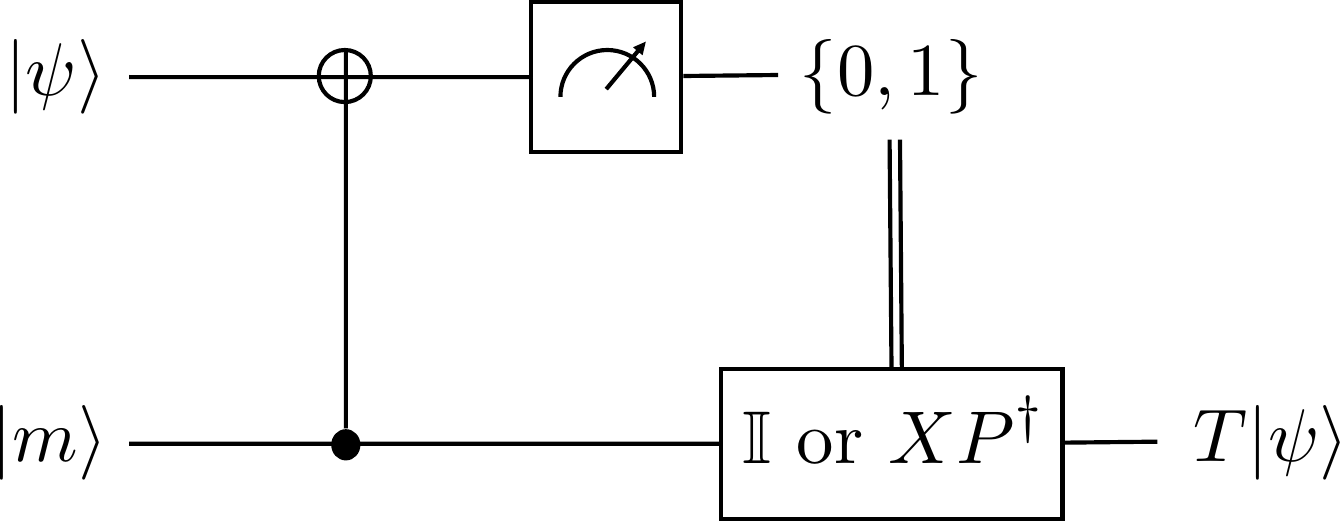}%
}{
\caption{Gate teleportation of the $T$ gate. The circuit applies the $T$ gate to an arbitrary single-qubit state $\rho$. Each state may be logical and each operation may be applied transversally. }\label{fig:gate_teleport}
}
\end{figure}

\onecolumngrid
\vspace{10em}
\begin{framed}
\begin{protocol}[Distributed Gate Teleportation (GTele)] \label{prot:gate_teleport}
~\\
\noindent Input: $\dbar \Psi$, $\ket{\bar{\bar{m}}}$ distributed by $D$ to the nodes and verified by the nodes using VQSS (Protocol \ref{prot:vhss}), set of apparent cheaters $B$ from verification of $\dbar \Psi$ and $\ket{\bar{\bar{m}}}$.\\
\noindent Output: Logical $T$ gate applied to the input logical state, $\dbar{T} (\dbar \Psi)$.

\begin{enumerate}
\item Each node $\ell$, for a share coming from node $i$:
\begin{enumerate}[label=(\alph*),topsep=0pt,itemsep=-1ex,partopsep=1ex,parsep=1ex]
\item applies CNOT with $\ket{\bar{\bar{m}}}_{i_\ell}$ as control qubit and  $\dbar \Psi_{i_\ell}$ as target qubit,
\item measures the target qubit in the $Z$ basis and broadcasts the result using the secure broadcast channel, see Network model.
\end{enumerate}
\item Broadcasted values yield words $\mathbf{v}_i$. Nodes publicly check on which positions the errors occurred using the classical decoder and update set $B$ with the positions of errors. They decode the classical value $a$: 
\begin{itemize}
\item If $a=0$, the nodes do not apply any correction.
\item If $a=1$, the nodes apply $XP^\dagger$ to their shares.
\end{itemize}
\end{enumerate}
\end{protocol}
\end{framed}
\vspace{2em}
\normalsize
\twocolumngrid

\subsubsection{Verification of Clifford-stabilized states}\label{subsec:verification_magic_state}
One last ingredient we need to perform distributed computation is to verify that the logical magic state $\ket{\dbar m}$ is indeed the logical magic state. This is necessary since we want to be sure that when we apply the $T$ gate in a distributed way, the result will be the $T$ gate on the shares of honest nodes. 

Here we present a protocol to verify the magic state in a distributed way. In fact, our protocol works for any qubit state $\ket{g}$ stabilized by a single-qubit Clifford gate $G$. Our idea is inspired by so called stabilizer measurement in quantum error correction, see Figure \ref{fig:verification_stab}. Consider a single-qubit gate $XP^\dagger$ with a $+1$ eigenstate $\ket{m}$.  Then it holds that the state $\ket{+}\ket{m}$ is stabilized by controlled $XP^\dagger$ gate, $C$-$XP^\dagger$, where $\ket{+}$ is used as a control and $\ket{m}$ is used as a target. That is,
\begin{align}
C\text{-}XP^\dagger (\ket{+}\ket{m}) = \ket{+}\ket{m}.
\end{align}
This gives us an insight into how the verification of $\ket{m}$ should work: if the target state was the magic state then after performing $C$-$XP^\dagger$ we will always measure the control in $\ket{+}$ (or equivalently, first apply $H$ and measure 0). Additionally, if the target was not in the magic state and we measure the control in $\ket{+}$, we will project the target onto $\ket{m}$. For this to work, one  needs to make sure that the control qubit was in $\ket{+}$ before applying the controlled gate.

We adapt this procedure to run on the logical level in a distributed way as follows. Using VQSS(0), the nodes first verify a logical $\ket{\dbar 0}$ encoded and shared by $D$. They also share $\ket{\dbar m}$ and verify that it is a valid codeword of $\ctr$ using the VQSS, Protocol \ref{prot:vhss}. This step is necessary since we want the transversal operations which the nodes will perform in next steps to be well defined. Each of the nodes now applies the Hadamard gate to her share of $\ket{\dbar 0}$ to turn it into a logical $\ket{\dbar +}$, and after that performs $C$-$XP^\dagger$ between her shares of $\ket{\dbar +}$ and $\ket{\dbar m}$. Then the nodes apply the Hadamard gate to the control qubits one more time and measure in the standard basis. They announce their measurement results and use the classical decoder to get the value $a$, just like in VQSS(0) and GTele. Note that the protocol works as long as the gate $C$-$XP^\dagger$ can be applied transversally with respect to the code used to encode $\ket{\dbar 0}$ and $\ket{\dbar m}$. 

Protocol \ref{prot:verification_stabilizer} requires an operational workspace of $4n$ qubits per node: first the verification of $\ket{\dbar m}$ requires a $3n$-qubit workspace per node. After this verification step, each node needs to store $n$ qubits of $\ket{\dbar m}$ and uses an extra $3n$-qubit workspace to verify $\ket{\dbar 0}$. This amounts to a $4n$-qubit workspace per node. The communication complexity is the same as in the sequential VQSS protocol, that is $\mathcal{O}(n^2s^2)$ qubits per node, where $s$ is the security parameter.

\begin{figure}[b]
\ffigbox{%
 \includegraphics[scale=0.3]{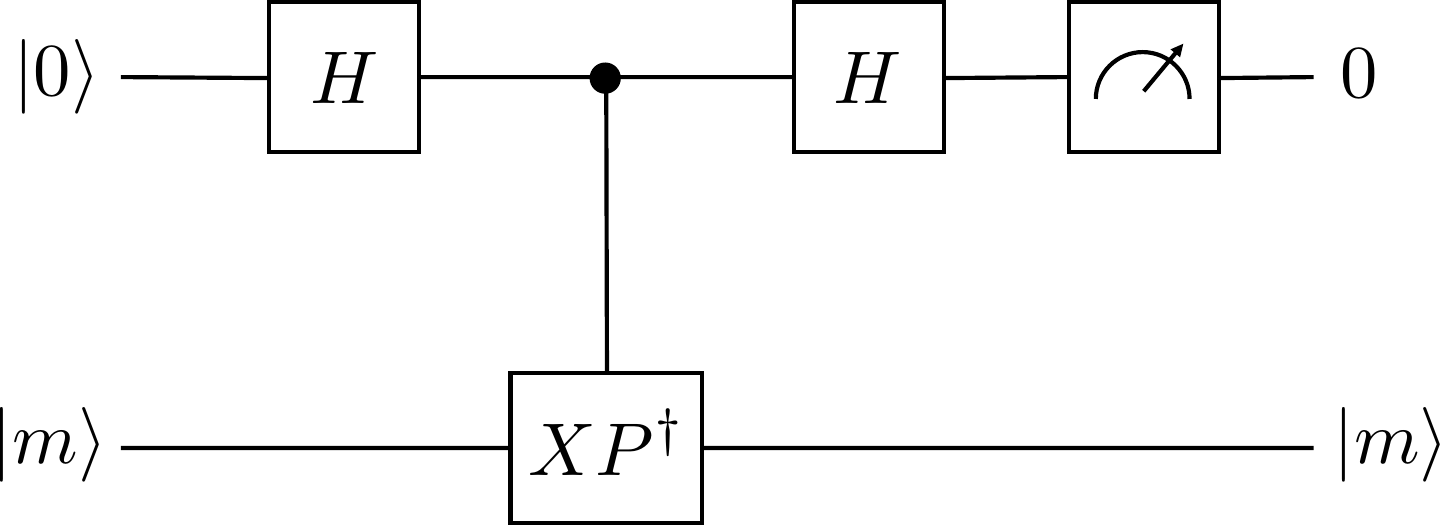}%
}{
\caption{Verification of the magic state using stabilizer measurement. The circuit verifies that the target input is the magic state using the fact that the state $\ket{+}\ket{m}$ is stabilized by the controlled $C$-$XP^\dagger$ gate. }\label{fig:verification_stab}
}
\end{figure}

\onecolumngrid
\vspace{2em}
\begin{framed}
\begin{protocol}[Verification of Clifford-Stabilized States (VMagic)] \label{prot:verification_stabilizer}
~\\
\noindent Input: $\ket{0}$ and $\ket{g}$ prepared by $D$, single-qubit Clifford gate $G$ stabilizing $\ket{g}$, error correcting code $\ctr$, set of apparent cheaters $B$. 
\noindent Output: verified logical states $\ket{\dbar{0}}$ and $\ket{\dbar{g}}$ 

\begin{enumerate}
\item The nodes run VQSS(0) with $\ket{0}$ as an input and VQSS with $\ket{g}$ as an input with dealer $D$. They update the set $B$ with apparent cheaters $B_{0}$ revealed in verifying $\ket{0}$ and apparent cheaters $B_g$ revealed in verifying $\ket{g}$. 
\item Each node $\ell$, for all shares coming from node $i$:
\begin{enumerate}[label=(\alph*),topsep=0pt,itemsep=-1ex,partopsep=1ex,parsep=1ex]
\item applies $H$ to $\ket{\bar{\bar{0}}}_{i_\ell}$,
\item applies $C$-$G$ with $\ket{\bar{\bar{0}}}_{i_\ell}$ as the control qubit and $\ket{\dbar{g}}_{i_\ell}$ as the target qubit,
\item applies $H$ to control qubit,
\item measures the control qubit in the $Z$ basis and broadcasts the result using the secure broadcast channel, see Network model.
\end{enumerate}
\item Broadcasted values yield words $\mathbf{v}_i$. Nodes publicly check on which positions the errors occurred using the classical decoder and update set $B$ with the positions of errors. They decode the classical value $a$: 
\begin{itemize}
\item If $a=0$, continue.
\item If $a=1$, set $B = [n]$ (this will cause the MPQC protocol to abort after the computation phase).
\end{itemize}
\end{enumerate}
\end{protocol}
\end{framed}
\vspace{2em}
\normalsize
\twocolumngrid

\subsection{Multi-party quantum computation }\label{subsec:MPQC}

We are now ready to perform a distributed computation using the ingredients from the previous sections. Recall, the goal of the protocol is to perform a circuit $\mathfrak{R}$ in a distributed way on $n$ single-qubit private inputs $\rho_1,\dots,\rho_n$, each coming from one node $1,\dots,n$. Note that the inputs can possibly be entangled. In universal MPQC we compute an arbitrary circuit $\mathfrak{R}$. We choose Clifford gates supplemented with a $T$ gate to be our universal set of gates. 

\paragraph*{Sharing and verification.} During this phase the nodes jointly verify whether dealer $D_i$ is honest, i.e. whether there are less than $t \leq \left\lfloor \frac{d-1}{2} \right\rfloor$ errors in the first-level encoding performed by $D_i$. They publicly record the positions on which the errors occurred in the set of apparent cheaters $B_{i}$ corresponding to dealer $D_i$. After all of the dealers are verified, they publicly construct a global set of apparent cheaters  $B$, see step 2 of Protocol \ref{prot:mpqc}. If $|B| \leq t$ the protocol continues. Note that $|B| \leq t$ implies that each of the honest nodes holds shares with at most $t$ errors on the same positions of the first level of encoding. Otherwise, when $|B|>t$, the honest nodes know they will abort the protocol after the computation and replace their shares with $\ket{0}$. This step is necessary to complete the security proof.

In this phase each node requires a workspace of $n^2 + 2n$ qubits to verify all of the inputs in a sequential way, and sends $(n+1)ns^2$ qubits, where $s$ is the security parameter. The size of the workspace for our MPQC protocol does not depend on $s$ since the verification phase of VQSS is performed in a sequential way.

\paragraph*{Computation.} In the computation phase, the goal is to compute the circuit $\mathfrak{R}$ on the twice-encoded (see Figure \ref{fig:twolevel_encoding}) and verified inputs.
Note that the set of $B$ of apparent cheaters created during the verification is public and cumulative throughout the protocol. That means that it accumulates errors from executions of VMagic, VQSS(0), GTele in the computation phase. If at any point $|B|>t$ during these protocols, the honest nodes proceed in the same way as in the verification phase -- they replace their shares with $\ket{0}$.
At the end of the computation phase the nodes look at the set $B$. If $|B| > t$ the protocol aborts. Otherwise, the nodes proceed to the reconstruction phase.

The inputs require a workspace of $n^2$ qubits per node. For applying the $T$ gate, each node needs a workspace of additional $4n$ qubits, see Protocol \ref{prot:verification_stabilizer}. Additionally, the verification of every ancilla in $\mathfrak{R}$ requires a workspace of $3n$ qubits per node. This means that each node requires a workspace of at most $n^2+4n$ qubits in total. In this phase, each node sends $\mathcal{O}\big(( \#\tn{ancillas} + \# T)ns^2 \big)$ qubits.

At this point the nodes hold a global state $\dbar \omega$. Let $\dbar \omega_k = \tr_{[n]\smallsetminus i} (\dbar \omega)$ be the outcome of each node~$i$.

\paragraph*{Reconstruction.} After the computation phase the cheating nodes can still introduce errors to the shares they hold before sending them back to corresponding dealers. Therefore, each of the dealers, after receiving her original shares back, runs an error correcting circuit for the code $\ctr$ and identifies further errors. If there is no more than $t$ errors, she reconstructs her output state~$\omega_i$.
In this phase, the nodes just exchange the existing qubits, therefore the operational workspace does not increase from $n^2 + 4n$. Each node sends $n^2$ qubits.

\onecolumngrid
\vspace{2em}
\begin{framed}
\begin{protocol}[Multi-party quantum computation (MPQC)] \label{prot:mpqc}
~\\
\noindent Input: private input $\rho_i$ for every node $i$, circuit $\mathfrak{R}$, error correcting code $\ctr$.
\\

\noindent\textbf{Sharing and Verification}
\begin{enumerate}
\item \label{step:sharing} Each node $i=1,\dots,n$ runs sequential verifiable quantum secret sharing (VQSS, Protocol \ref{prot:vhss}) with single-qubit input $\rho_i$ and code $\ctr$, acting as dealer $D_i$. This way nodes create logical $\dbar \Psi_i$ encoded twice with $\ctr$, see Figure \ref{fig:twolevel_encoding}.
\item The nodes publicly create sets $B_{i,\ell}$ containing all second-level errors from all $n$ executions of sequential VQSS (see \cite{Crepeau2002,Lipinska2019} for details). If for each node $\ell$, if $|B_{i,\ell}| > t$ then they add node $\ell$ to a set of apparent cheaters $B_i$ for dealer $D_i$. After all $n$ executions of VQSS, they create a global set of apparent cheaters $B = \bigcup_i B_i$. If $|B| > t$ the nodes know they will abort after the computation. They replace all the shares they hold with $\ket{0}$.
\end{enumerate}
\textbf{Computation}
\begin{enumerate}[resume]
\item For every Clifford gate $C$ of the circuit $\mathfrak{R}$ the nodes apply $C$ transversally to their local qubits. For every $T$ gate in $\mathfrak{R}$ applied to the input of $D_i$:
\begin{enumerate}[label=(\alph*),topsep=0pt,itemsep=-1ex,partopsep=1ex,parsep=1ex]
\item $D_i$ creates $\ket{0}$ and $\ket{m}$. The nodes run Verification of Clifford-Stabilized States (VMagic, Protocol \ref{prot:verification_stabilizer}). The nodes update the set $B$ with apparent cheaters from execution of VMagic. If $|B| > t$ the nodes replace all the shares they hold with $\ket{0}$.
\item The nodes apply Distributed Gate Teleportation (GTele, Protocol \ref{prot:gate_teleport}) to their shares of $\dbar \Psi_i$ and verified $\ket{\dbar m}$. The nodes update the set $B$ with apparent cheaters from execution of GTele. If $|B| > t$ the nodes replace all the shares they hold with $\ket{0}$ and do not apply a correction in GTele (treating the measurement outcome as 0).
\end{enumerate}
\item For every $\ket{0}$ ancilla necessary to perform the circuit $\mathfrak{R}$, a node $i \notin B$ chosen at random using the public source of randomness, runs VQSS(0) acting as a dealer. They update $B$ with the set of apparent cheaters from the execution of VQSS(0). The nodes use the verified $\ket{\dbar 0}$ to perform $\mathfrak{R}$. If $|B| > t$ the nodes replace all the shares they hold with $\ket{0}$.
\item If $|B|>t$ the protocol aborts. Otherwise continue.
\end{enumerate}
Let the logical global outcome of the computation be $\dbar{\omega}$, with $\dbar \omega_i = \tr_{[n]\smallsetminus i} (\dbar \omega)$ corresponding to the outcome of each node~$i$.\\

\noindent\textbf{Reconstruction}
\begin{enumerate}[resume]
\item Each node sends all of the shares of $\dbar \omega_i$ to $D_i$. 
\item Each $D_i$:
\begin{enumerate}[label=(\alph*),topsep=0pt,itemsep=-1ex,partopsep=1ex,parsep=1ex]
\item  For each share coming from node $j \notin B$, $D_i$ runs an error correcting circuit for the code $\ctr$. She creates a set of errors $\tilde B_{i,j}$ such that it contains $B_{i,j}$, i.e. $B_{i,j} \subseteq \tilde B_{i,j}$. If $|\tilde B_{i,i}| \leq t$ then errors are correctable, $D_i$ corrects them and decodes the $i$-th share obtaining $\bar \omega_i$. Otherwise, $D_i$ adds $j$ to the global set $B$.
\item For all $j \notin B$, $D_i$ randomly chooses $n-2t$ shares of  $\bar \omega_i$ and applies an erasure-recovery circuit to them. She obtains $\omega_i$.
\end{enumerate}
\end{enumerate}
\end{protocol}
\end{framed}
\vspace{2em}
\normalsize
\twocolumngrid

Altogether, each node requires an operational workspace of $n^2+4n$ qubits, and sends $\mathcal{O}\big((n + \#\tn{ancillas} + \# T)ns^2 \big)$ qubits throughout the execution of the MPQC protocol, Protocol \ref{prot:mpqc}.

\subsection{Security statements} \label{subsec:Security}

In this section we prove the security of our MPQC protocol. To do so, we first state the security framework and definition following the work of \cite{Beaver1992,Micali1992,Canetti2001,Unruh2010}. We employ the simulator-based security definition, see Definition \ref{def:security} below. It implies that the three properties -- correctness, soundness and privacy defined at the beginning of this manuscript -- are automatically satisfied.
Our security definition uses two models of the protocol -- ``real'' and ``ideal''.  The real model corresponds to the execution of the actual MPQC protocol. In the ideal model the nodes interact with an oracle that perfectly realizes the MPQC task and is incorruptible. The general idea is that the protocol is secure if one cannot distinguish a real execution of MPQC from the ideal one.

In the ideal model the honest nodes can only interact with the oracle. What is more, they do so in a so called ``dummy'' way, i.e. they simply forward their input to the oracle, and output whatever they receive from the oracle. The cheating nodes can collude and perform any joint operation on their inputs before sending it to the oracle. Similarly, they can perform any joint operation on whatever they receive from the oracle before they output their state. Recall that we do not make any assumption on the computational power of the cheaters. For the purpose of the proof we will say that the cheaters can be corrupted by an adversary $\mathcal{A}$ which can corrupt at most $t$ nodes, but otherwise is arbitrarily powerful.
Moreover, by $\mathcal{A}_{\rm real}$ we will denote the adversary in the ``real'' protocol and by $\mathcal{A}_{\rm ideal}$ the adversary in the ``ideal'' protocol. 

\begin{definition}[$\epsilon$-security]\label{def:security}
We say that a MPQC protocol $\Pi$ is $\epsilon$-secure if for any input state $\rho$, and any real adversary $\mathcal{A}_{\rm real}$, there exists an ideal adversary $\mathcal{A}_{\rm ideal}$, such that the output state $\omega_{\rm real}:=\Pi_{\rm real}(\rho)$ of the real protocol is $\epsilon$-close to the output state $\omega_{\rm ideal}:=\Pi_{\rm ideal}(\rho)$ of the ideal protocol, that is
\begin{align}
\tfrac{1}{2} \parallel \omega_{\rm real} - \omega_{\rm ideal}   \parallel_1 \leq \epsilon.
\end{align}
\end{definition}
To prove the security of the MPQC protocol, Protocol \ref{prot:mpqc} we first restate the soundness of the VQSS protocol \cite{Crepeau2002,Smith2001,Lipinska2019}.

\begin{lemma}[Soundness of VQSS]\label{lem:vqss_sound}
In the verifiable quantum secret sharing protocol, Protocol \ref{prot:vhss}, either the honest parties hold a consistently encoded secret or the dealer is caught with probability at least $1-2^{-\Omega(s)}$.
\end{lemma}

\begin{restatable}[]{thm}{securityMPQC}
\label{thm:security}
 	The Multi-party Quantum Computation protocol, Protocol \ref{prot:mpqc}, is $\kappa 2^{-\Omega(s)}$-secure, where $\kappa = n + \# T$ gates $+ \#$ ancillas in $\mathfrak{R}$.
\end{restatable}

\begin{proof}[Idea of the proof]
Our proof is inspired by the approach taken in \cite{Crepeau2002,Smith2001}, on which we expand and explicitly show that the outputs of the real and ideal protocol are $\epsilon$-close, see Appendix \ref{app:proofs}. We construct an ideal protocol using a common simulation technique, where $\mathcal{A}_{\rm ideal}$ locally simulates the MPQC protocol, Protocol \ref{prot:mpqc}, with honest nodes interacting with the cheaters. This means that for any real adversary $\mathcal{A_{\rm real } }$ we construct an ideal adversary $\mathcal{A}_{\rm ideal}$ by saying that it internally simulates the execution of real protocol with the real adversary $\mathcal{A_{\rm real } }$. 
Then we formally write the execution of the real protocol. We show that the outputs of both protocols are equal in the case when the encoding in the sharing phase of Protocol \ref{prot:mpqc} is done correctly. We also prove that the $\epsilon$ error in the security comes from the fact that the verification of inputs and any ancillas needed for MPQC can fail with probability defined by Lemma \ref{lem:vqss_sound}.
\end{proof}

We remark that our security definition follows the paradigm of sequential composability, formalized by the real-vs-ideal security definition, Definition \ref{def:security}. The extendibility of our security definition to the more general framework of universal composability \cite{Canetti2001,Unruh2010} is left as an open problem.

\section{Discussion} \label{sec:outlook}
In our protocol we allow an abort when there are too many errors introduced by the cheaters, see Protocol \ref{prot:mpqc}. However, this condition can be removed following the approach outlined in \cite{Crepeau2002,Smith2001} (there called Top-Level Sharing), at the cost of more rounds of quantum communication. Given our objective is to save resources, we did not pursue this path in this manuscript. However, we can introduce a step before computation, in which the nodes perform a distributed encoding (creating the third level of encoding) of the verified inputs. It works as follows. The nodes run the VQSS verification procedure for every input state $\rho_i$, but do not create a global set of cheaters. Instead, they create a set $B_i$ recording first-level errors on input state $\rho_i$.
To perform the distributed encoding of input $\rho_i$ the nodes use ancilla states prepared and encoded by the corresponding dealer $D_i$. The nodes also verify the ancillas using VQSS and add the errors that occurred on the first level of encoding of ancillas to $B_i$. If $|B_i| \leq t$, the nodes perform the distributed encoding with the verified ancillas.
The encoding can be done transversally, since for any stabilizer error correcting code the encoding procedure is a Clifford operation \cite{Gottesman2009}. 

If a dealer is caught cheating, $|B_i| > t$, the protocol does not abort. Instead, a node which has not been caught cheating yet prepares an encoding of $\ket{0}$ and the nodes proceed to verify it in the same way as before. Note that there will be at most $t$ failed tries in preparing a valid encoding of $\ket{0}$ since there are at most $t$ cheaters. Otherwise, upon a successful verification of the encoded $\ket{0}$, the nodes proceed to the distributed encoding. This step replaces the invalid input from the cheater with a valid encoding of $\ket{0}$.
Same procedure, ``try until you succeed'', can be adapted to verify magic states and $\ket{0}$ ancillas needed to perform the circuit $\mathfrak{R}$. The nodes simply try until the verification of an ancilla has at most $t$ errors. 

Performing the distributed encoding of the inputs creates a three-level encoding before the computation phase. The shares initially dealt by dealer $D_i$ are then sent back to $D_i$, who reconstructs them and corrects the errors using the reconstruction step from VQSS (as in reconstruction of MPQC, Protocol \ref{prot:mpqc}). As a result, each node holds a single qubit corresponding to a correctly encoded input state $\rho_i$, with at most $t$ errors confined to the cheaters' positions. The protocol proceeds with the distributed computation, but now the circuit is performed on a single level of encoding. Since the errors are only on the shares held by the cheaters, the errors will not propagate to the honest shares during the computation. Therefore, after the computation it will be possible to reconstruct outputs for the honest nodes. 

Finally, we remark that the distributed encoding can be performed in a sequential way, similar to the execution of VQSS we present in Protocol \ref{prot:vhss}. In fact, this does not increase the qubit workspace per node, each node will not exceed the workspace of $n^2+4n$. However, this approach has significantly higher quantum communication complexity. Specifically, in this version of the protocol, each node needs to send~$\mathcal{O}(n^5s^2)$ qubits.

\section{Acknowledgments}
We thank B.~Dirkse for useful discussions, and K.~Chakraborty and M.~Skrzypczyk for detailed feedback on this manuscript. This work was supported by an NWO VIDI grant, an ERC Starting grant and NWO Zwaartekracht QSC. This project (QIA) has received funding from the European Union’s Horizon 2020 research and innovation program under grant agreement No 820445.

%

\onecolumngrid
\begin{appendix}

\section{Security proof.}\label{app:proofs}

Here we provide the security proof of our protocol based on the simulator definition, see Definition \ref{def:security}. We first construct the ideal protocol step by step and model each operation performed in this protocol by general maps, and finally express the output of this protocol $\omega_{ideal}$ in terms of these maps. Then, we analyze the real protocol and similarly express its output $\omega_{real}$ in terms of the maps modeling the real protocol. Finally, we compare the two outputs, $\omega_{ideal}$ and $\omega_{real}$, and show they are exponentially close in the security parameter $s$.

To prove security of the MPQC protocol, Theorem \cref{thm:security}, we first state the following useful lemma. Intuitively, it says that sharing and verifying the input, performing the distributed circuit and decoding is equivalent to applying the circuit to the inputs directly. Note that we consider the decoding to be ``hypothetical'' -- after the computation phase in MPQC the nodes send all of the shares coming from input of node $i$ to node $i$, and node $i$ reconstructs it.

\begin{lemma}\label{lem:conserv_enc}
Let $B$ be a set of apparent cheaters at the end of the computation phase, such that $|B|\leq t$, and $A$ be a set of cheaters. Let $\mathcal{D}$ denote the decoding procedure for code $\hat{\mathcal{C}}$ and $\hat{\mathcal{D}}$ denote the erasure recovery circuit for code $\hat{\mathcal{C}}$.  If the state $\dbar \rho$ encoded twice with the code $\hat{\mathcal{C}}$ is decodable, i.e.
\begin{align}
 \rho = \bigotimes_{i \in [n]} \left(\hat{\mathcal{D}}_{\overline{B\cup A}} \circ \bigotimes_{\ell \in \overline{B\cup A}} \mathcal{D}_\ell \right) (\dbar \rho),
\end{align}
then applying a logical gate $\dbar G$ ($G\in \tn{Cliff} +T$) on $\dbar \rho$ is also decodable, i.e.
\begin{align}
 G(\rho) = \bigotimes_{i \in [n]} \left(\hat{\mathcal{D}}_{\overline{B\cup A}} \circ \bigotimes_{\ell \in \overline{B\cup A}} \mathcal{D}_\ell \right) \left( \dbar G (\dbar \rho) \right),
\end{align}
where $\dbar G$ is gate $G$ applied transversally on the CSS code $\hat{\mathcal{C}}$ if $G\in \tn{Cliff}$, or it is the implementation of the $T$ gate described in Protocol \ref{prot:gate_teleport} if $G=T$.
The same property holds when replacing $G$ by the projective measurement in the $Z$ basis denoted $P$, and where $\dbar P$  corresponds to measuring each qubit of the double-encoded state in the $Z$ basis followed by broadcasting the outcome classically.
\end{lemma}
\begin{proof}
The lemma follows from the fact that to realize a logical gate $\dbar G$ it is sufficient to apply $G$ honestly on shares in $\overline{B\cup A}$. Indeed, applying a Clifford gate transversally on shares in $\overline{B\cup A}$ realizes a logical Clifford gate \cite{Gottesman1998}. For a CSS code $\hat{\mathcal{C}}$ measuring each qubit in the $Z$ basis and broadcasting the measurement result realizes the logical transversal measurement. Additionally, we implement the $T$ gate by composing an ancilla state, $Z$ measurement and a Clifford operation. Therefore, the transversal properties of Cliffords and $Z$ measurement can be transferred to this implementation of the $T$ gate.
\end{proof}

\begin{property}\label{prop:prop1}
Let $\mathfrak{R}$ be a circuit implementing a completely positive trace preserving (CPTP) map. Lemma \ref{lem:conserv_enc} holds when replacing $G$ by any circuit $\mathfrak{R}$,
\begin{align}
 \mathfrak{R}(\rho) = \bigotimes_{i \in [n]} \left(\hat{\mathcal{D}}_{\overline{B\cup \bar{H}}} \circ \bigotimes_{\ell \in \overline{B\cup \bar{H}}} \mathcal{D}_\ell \right) \left( \dbar{\mathfrak{R}} (\dbar \rho) \right).
\end{align}
This follows from the fact that any circuit $\mathfrak{R}$ can be represented as $\mathfrak{R} = P \circ \mathcal{U}$, where $\mathcal{U}$ can be decomposed into gates from the set Cliff$+T$ and $P$ is a measurement. 
\end{property}

As a reminder, let us restate the security of our MPQC protocol.

\securityMPQC*

\begin{proof}[Proof of \cref{thm:security}] 

This proof is inspired by the approach taken in \cite{Crepeau2002,Smith2001}. In the following we construct a proof aiming to show that the outputs of the real and ideal protocol are $\epsilon$-close. We first construct an ideal protocol using a simulator approach and formally state every step of the simulation. Then we formally write the execution of the real protocol.

\begin{framed}
 \noindent Box 1. Registers used in the security proof. \\

\noindent \textit{Ideal protocol:}\\
$H_S$ --  registers of ``dummy'' inputs of the honest nodes in the simulation\\
$A_S$ -- registers of the cheaters' inputs\\
$H_0$ --  registers of the simulated honest nodes \\
$A_0$ -- registers of the simulated cheaters.\\

\noindent \textit{Real protocol:}\\
$H_R$ --  registers of honest nodes \\
$A_R$ -- registers of cheaters.
\end{framed}

\begin{figure}[h]
\ffigbox{%
 \includegraphics[width=0.8\textwidth]{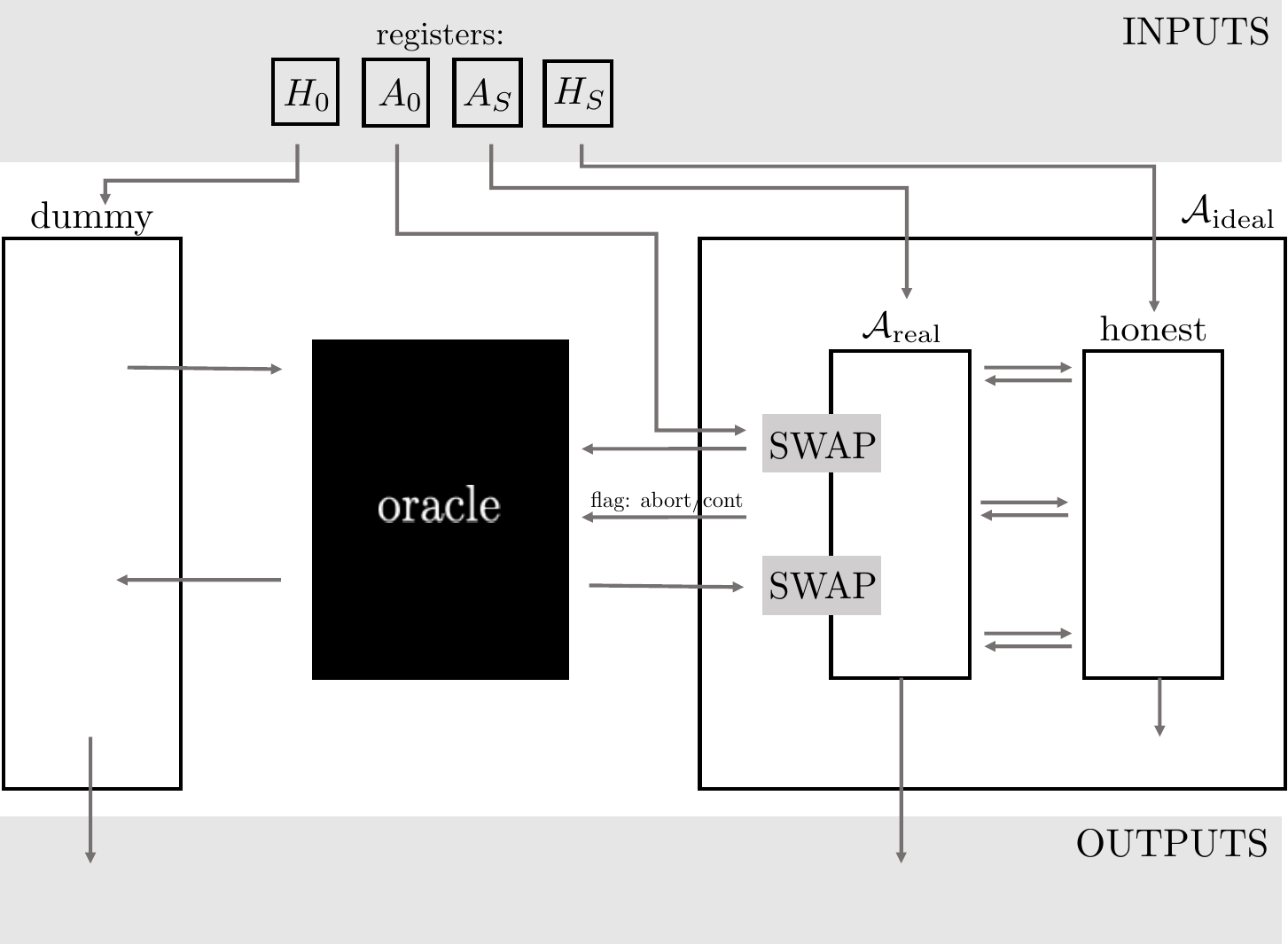}%
}{
\caption{Schematic of our simulator-based security proof of the MPQC protocol, Protocol \ref{prot:mpqc}.}\label{fig:simulator}
}
\end{figure}

\vspace{1em}
\noindent \textbf{Ideal protocol.} $\mathcal{A}_{\rm ideal}$ will locally simulate the MPQC protocol, Protocol \ref{prot:mpqc}, with honest nodes interacting with the cheaters. The cheaters are  controlled by $\mathcal{A}_{\rm real}$ and $\mathcal{A}_{\rm real}$ is simulated within $\mathcal{A}_{\rm ideal}$, see Figure \ref{fig:simulator}. In the ideal model $\mathcal{A}_{\rm ideal}$ and the honest nodes interact with an oracle that perfectly realizes the MPQC task and is incorruptible. The oracle requires two types of inputs: first the input registers $H_S,A_0$ on which the computation of the circuit will occur, second a flag input that indicates whether the oracle should abort or continue. If the flag input is ``abort'' the oracle outputs $\ketbra{\perp}$. If the the flag input is ``continue'' the oracle outputs the evaluation of circuit $\mathfrak{R}$ on the inputs $H_S A_0$. At any moment of this simulated execution, the ideal adversary has access to all the simulated registers, in particular, the set $B$ of apparent cheaters.
Let the input to the simulation be 
\begin{align}
\rho_{H_SA_S} \otimes \ketbra{0}_{H_0A_0} ,
\end{align}
where $\rho_{H_SA_S} $ denotes the input state of all nodes, such that $\tr_{[n] \setminus i}\left(\rho_{H_SA_S}\right) = \rho_i$.

1. $\mathcal{A}_{\rm ideal}$ locally simulates sharing and verification with simulated honest nodes using $\ket{0}$ as their input. The input registers $H_0A_S$ given to $\mathcal{A}_{\rm ideal}$ is forwarded to the simulated $\mathcal{A}_{\rm real}$, i.e.

\begin{align} \label{eq:SV_on_input}
\sigma^{(1)}_{H_0A_0H_SA_S} = 
 \mathcal{SV}_{H_0 A_S}
\left( \rho_{H_SA_S} \otimes \ketbra{0}_{H_0A_0}\right),
\end{align}
where $\mathcal{SV}_{H_0 A_S}$ denotes the sharing and verification (see Protocol \ref{prot:mpqc}) performed on registers $H_0$ and $A_S$. We assume that the identity operation is applied on all the registers that are not in the map  $\mathcal{SV}$, i.e. $\mathds{1}_{H_S A_0}$.

2. Before $\mathcal{A}_{\rm ideal}$ proceeds with the simulation of the computation phase, for each input of the cheaters $\mathcal{A}_{\rm ideal}$, creates an encoding of $\ket{0}$ in register $A_0$. Then $\mathcal{A}_{\rm ideal}$ performs a swap gate between $A_0$ and cheaters' input $A_S$.
\begin{itemize}
\item In the case when the set $|B| \leq t$, there are sufficiently few errors on both levels of encoding. Then $\mathcal{A}_{\rm ideal}$ can apply an erasure-recovery circuit twice (for the double encoding), denoted $\tilde{\mathcal{D}}_{A_0}$, to the input of nodes not in $B$ and pass it to the oracle. Applying decoding $\tilde{\mathcal{D}}_{A_0}$ is necessary, since the oracle accepts only single-qubit inputs. 
\item Otherwise, when $|B| > t$, $\mathcal{A}_{\rm ideal}$ simply passes previously prepared $\ket{0}$ states as inputs of the cheaters to the oracle and the simulated honest nodes $H_S$ replace their shares with $\ket{0}$. The simulated cheaters apply an arbitrary map $\mathcal{M}_{A_S}$ to their shares.
\end{itemize}
We therefore describe this step as
\begin{align}
\sigma^{(2)}_{H_0A_0H_SA_S} = 
\begin{cases}
 \tilde{\mathcal{D}}_{A_0} \circ \tn{Swap}_{A_0A_S} \circ  \mathcal{E}_{A_0}  (\sigma^{(1)}_{H_0A_0H_SA_S})  &\tn{~if~} |B| \leq t\\
 \mathcal{M}_{A_S}\otimes\tr_{H_0}[\sigma^{(1)}_{H_0A_0H_SA_S}] \otimes \ketbra{0}_{H_0} & \tn{~if~} |B| > t.
\end{cases}
\end{align}

3. $\mathcal{A}_{\rm ideal}$ proceeds with the simulation of the computation phase on registers $H_0$ and $A_S$. At the same time, the oracle computes the ideal circuit $\mathfrak{R}^{\rm ideal}_{H_S A_0}$ on the simulated honest shares $H_S$ and register $A_0$ of the cheaters. The state after this step is therefore,
\begin{align} \label{eq:ideal_step2}
\sigma^{(3)}_{H_0A_0H_SA_S} = 
\begin{cases}
 (\mathfrak{R}^{\rm ideal}_{H_S A_0} \otimes \dbar{\mathfrak{R}}_{H_0 A_S}) (\sigma^{(2)}_{H_0A_0H_SA_S})  &\tn{~if~} |B| \leq t\\
 (\mathfrak{R}^{\rm ideal}_{H_S A_0} \otimes \dbar{\mathfrak{R}}_{H_0 A_S}) (\sigma^{(2)}_{H_0A_0H_SA_S}) &\tn{~if~} |B| > t.
\end{cases}
\end{align}

4. If $|B| > t$, $\mathcal{A}_{\rm ideal}$  sends the flag ``abort'' to the oracle, and ``continue'' otherwise. 
\begin{itemize}
\item If the oracle receives ``abort'' it outputs a flag $\ketbra{\perp}$ to all nodes. 
\item Otherwise, it outputs the computation of the ideal circuit on the inputs.
\end{itemize}

5. The nodes in $H_S$ output whatever they received from the oracle. Upon receiving the oracle's output, $\mathcal{A}_{\rm ideal}$ does the following: 
\begin{itemize}
\item if ``abort'' was sent in the previous step, then it must be that $|B|>t$. The simulated protocol aborts. Therefore, $\mathcal{A}_{\rm ideal}$ outputs the output of the $\mathcal{A}_{\rm real}$. Note that the simulated cheaters could have applied an arbitrary map $\mathcal{M}'_{A_S}$ on their register.
\item if ``continue'' sent in the previous step, then $\mathcal{A}_{\rm ideal}$ applies double encoding $\mathcal{E}_{A_0}$ to all shares of the cheating nodes $A_0$. Then $\mathcal{A}_{\rm ideal}$ applies the swap gate between the simulated registers of cheaters $A_S$ and $A_0$, and proceeds to the next step.
\end{itemize} 
\begin{align}\label{eq:ideal_step5}
\begin{cases}
 \tn{Swap}_{A_0A_S} \circ \mathcal{E}_{A_0} \circ (\mathfrak{R}^{\rm ideal}_{H_S A_0} \otimes \dbar{\mathfrak{R}}_{H_0 A_S})   (\sigma^{(2)}_{H_0A_0H_SA_S}) \otimes \ketbra{\tn{cont}}  &\tn{~if~} |B| \leq t\\
\ketbra{\perp}_{H_SA_0} \otimes \tr_{H_SA_0}\left[ \mathcal{M}'_{A_S} (\sigma^{(3)}_{H_0A_0H_SA_S}) \right]  \otimes \ketbra{\tn{abort}} &\tn{~if~} |B| > t.
\end{cases}
\end{align}
Let us denote by $\sigma^{(5)}_{H_0A_0H_SA_S}$ the following and use the explicit form of $\sigma^{(2)}_{H_0A_0H_SA_S}$ for $|B| \leq t$, Equation \eqref{eq:ideal_step2},
\begin{align}
\sigma^{(5)}_{H_0A_0H_SA_S} =  \tn{Swap}_{A_0A_S} \circ \mathcal{E}_{A_0} \circ (\mathfrak{R}^{\rm ideal}_{H_S A_0} \otimes \dbar{\mathfrak{R}}_{H_0 A_S}) \circ \tilde{\mathcal{D}}_{A_0} \circ \tn{Swap}_{A_0A_S} \circ  \mathcal{E}_{A_0} (\sigma^{(1)}_{H_0A_0H_SA_S}).
\end{align}
We will now simplify the above expression. For this we first state the following useful property.
 
\vspace{1em}
\begin{property}
For any operation $\mathcal{O}_{ABCD}$ on registers $ABCD$, the following identity holds,
\begin{align}
\tn{Swap}_{BC} \circ \mathcal{O}_{ABCD} \circ \tn{Swap}_{BC} = \mathcal{O}_{ACBD}.
\end{align}
\end{property}

\vspace{1em}
Using this property for $\sigma^{(5)}_{H_0A_0H_SA_S}$ we get that 
\begin{align}
\tn{Swap}_{A_0A_S} \circ \mathcal{E}_{A_0} \circ (\mathfrak{R}^{\rm ideal}_{H_S A_0} \otimes \dbar{\mathfrak{R}}_{H_0 A_S}) \circ \tilde{\mathcal{D}}_{A_0} \circ \tn{Swap}_{A_0A_S} \circ  \mathcal{E}_{A_0}
=
 \mathcal{E}_{A_S} \circ \mathfrak{R}^{\rm ideal}_{H_S A_S} \circ \tilde{\mathcal{D}}_{A_S} \otimes \dbar{\mathfrak{R}}_{H_0 A_0} \circ  \mathcal{E}_{A_0}.
\end{align}
This means that that the composition of the swaps with the ideal circuit performed by the oracle is equivalent to applying the ideal circuit to registers $H_S A_S$ by the oracle.
Therefore, we can simplify $\sigma^{(5)}_{H_0A_0H_SA_S}$ to
\begin{align}
\sigma^{(5)}_{H_0A_0H_SA_S}= (\mathcal{E}_{A_S} \circ \mathfrak{R}^{\rm ideal}_{H_S A_S} \circ \tilde{\mathcal{D}}_{A_S}) \otimes (\dbar{\mathfrak{R}}_{H_0 A_0} \circ  \mathcal{E}_{A_0} )(\sigma^{(1)}_{H_0A_0H_SA_S}).,
\end{align}
and using Equation \eqref{eq:SV_on_input} we obtain,
\begin{align}
 \sigma^{(5)}_{H_0A_0H_SA_S} & = (\mathcal{E}_{A_S} \circ \mathfrak{R}^{\rm ideal}_{H_S A_S} \circ \tilde{\mathcal{D}}_{A_S}) \otimes (\dbar{\mathfrak{R}}_{H_0 A_0} \circ  \mathcal{E}_{A_0} \circ \mathcal{SV}_{H_0A_S}) \left( \rho_{H_SA_S} \otimes \ketbra{0}_{H_0A_0}\right) \\
 &\label{eq:sigma_4} = \left( \mathcal{E}_{A_S} \circ \mathfrak{R}^{\rm ideal}_{H_S A_S}\circ \tilde{\mathcal{D}}_{A_S} \circ \mathcal{SV}_{A_S} \left( \rho_{H_SA_S} \right) \right) \otimes \left( \dbar{\mathfrak{R}}_{H_0 A_0} \circ \mathcal{E}_{A_0} \circ \mathcal{SV}_{H_0} \left( \ketbra{0}_{H_0A_0}\right) \right).
\end{align}

6. If the protocol did not abort, $\mathcal{A}_{\rm ideal}$ proceeds to the reconstruction phase, in which the simulated honest nodes $H_0$ first use the decoding procedure for code $\hat{\mathcal{C}}$ and then apply an erasure recovery circuit, as in the reconstruction phase of Protocol \ref{prot:mpqc}. We denote this procedure by $\tilde{\mathcal{D}}_{H_0}$. On the other hand, the simulated cheaters $A_S$ apply an arbitrary map $\mathcal{W}_{A_S}$. $\mathcal{A}_{\rm ideal}$ outputs whatever is the output of the simulated $\mathcal{A}_{\rm real}$. Therefore, the output of the ideal protocol is
\begin{align}
\omega_{\rm ideal } = \tr_{H_0 A_0} \left[\tilde{ \mathcal{D}}_{H_0} \otimes \mathcal{W}_{A_S} (\sigma^{(5)}_{H_0A_0H_SA_S})\right].
\end{align}
Using Equation \eqref{eq:sigma_4} and the fact that the sharing and verification followed double decoding, $\tilde{\mathcal{D}}_{A_S} \circ \mathcal{SV}_{A_S}$, is equivalent to $\mathds{1}_{A_S}$, we obtain,
\begin{align}
\omega_{\rm ideal } = \mathcal{W}_{A_S} \circ \mathcal{E}_{A_S} \circ \mathfrak{R}^{\rm ideal}_{H_S A_S}\left( \rho_{H_SA_S} \right).
\end{align}
Similarly, to later compare with the real protocol, we write the identity map on $H_S$ as $\mathds{1}_{A_S} = \tilde{\mathcal{D}}_{H_S}\circ \mathcal{E}_{H_S}$, and get
\begin{align} \label{eq:omega_ideal}
\omega_{\rm ideal } = (\tilde{\mathcal{D}}_{H_S}\otimes \mathcal{W}_{A_S}) \circ \mathcal{E}_{H_SA_S} \circ \mathfrak{R}^{\rm ideal}_{H_S A_S}\left( \rho_{H_SA_S} \right).
\end{align}

\vspace{2em}
\noindent \textbf{Real protocol.} 

In the real protocol whenever the honest nodes observe $|B|>t$ they replace all of their shares with $\ket{0}$. This is necessary because in the ideal protocol the oracle receives ``abort'' at the end of the computation phase. Therefore, in the real protocol we also abort at the end of the computation phase. However, it could happen that in the case when $|B|>t$ continuing the computation leaks some information about the honest nodes' inputs. To avoid this situation, we make the honest nodes substitute their shares with $\ket{0}$.

1. The protocol starts with the sharing and verification phase, which we describe by the map $\mathcal{SV}$ acting on inputs of all the nodes $\rho_{H_RA_R}$.  The state after this step is
\begin{align}
\mathcal{SV}_{ H_RA_R} \left( \rho_{H_RA_R}\right).
\end{align}

2. The protocol continues; 
\begin{itemize}
\item In the case when $|B| \leq t$, the nodes apply the distributed circuit $\dbar{\mathfrak{R}}_{ H_RA_R}$.
\item In the case when $|B| > t$, the honest nodes replace their shares with $\ket{0}$ and the cheaters apply an arbitrary map $\mathcal{M}_{A_R}$.
\end{itemize}
  At the end of the computation phase the state is therefore, 
\begin{align}
\sigma^{(2)}_{H_RA_R}=
\begin{cases}
\dbar{\mathfrak{R}}_{ H_RA_R} \circ \mathcal{SV}_{ H_RA_R} \left( \rho_{H_RA_R}\right) &\tn{~if~} |B| \leq t,\\
\mathcal{M}_{A_R}(\tr_{H_R}[\mathcal{SV}_{ H_RA_R} \left( \rho_{H_RA_R}\right)]) \otimes \ketbra{0}_{H_R} & \tn{~if~} |B| > t.
\end{cases}
\end{align}

3. The nodes check the size of set $B$. 
\begin{itemize}
\item If $|B| \leq t $ then the protocol continues to the reconstruction phase, where the honest nodes apply first a decoding operator for code $\hat{\mathcal{C}}$ and then an interpolation circuit, denoted $\mathcal{D}_{H_R}$. At the same time, the cheaters can apply an arbitrary map on their registers, which we denote $\mathcal{W}_{A_R}$. 
\item In the case when $|B| > t$, the nodes output the abort flag $\ketbra{\perp}$ and the cheaters output their part of $\sigma^{(2)}_{H_RA_R}$, possibly with an arbitrary map $\mathcal{M}'_{A_R}$. The protocol aborts. 
\end{itemize}
We can describe this step as,
\begin{align}\label{eq:real_step3}
\begin{cases}
(\mathcal{D}_{H_R} \otimes \mathcal{W}_{A_R})\circ \dbar{\mathfrak{R}}_{ H_RA_R} \circ \mathcal{SV}_{ H_RA_R} \left( \rho_{H_RA_R}\right) \otimes \ketbra{\tn{cont}} &\tn{~if~} |B| \leq t,\\
\ketbra{\perp}_{H_R} \otimes \mathcal{M}'_{A_R}(\tr_{H_R}[\sigma^{(2)}_{H_RA_R}]) \otimes \ketbra{\tn{abort}}&\tn{~if~} |B| > t.
\end{cases}
\end{align}
We introduce the identity map as encoding followed by double encoding on both registers, i.e. $\mathds{1}_{H_RA_R} = \tilde{\mathcal{D}}_{H_RA_R} \circ \mathcal{E}_{H_RA_R}$. Then, plugging this $\mathds{1}_{H_RA_R}$ between $\dbar{\mathfrak{R}}_{H_RA_R}$ and $(\mathcal{D}_{H_R} \otimes \mathcal{W}_{A_R})$, the first case can be rewritten as, 
\begin{align}
\omega_{\tn{real}} = (\mathcal{D}_{H_R} \otimes \mathcal{W}_{A_R})\circ \mathcal{E}_{H_RA_R} \circ \tilde{\mathcal{D}}_{H_RA_R} \circ \dbar{\mathfrak{R}}_{H_RA_R}\circ \mathcal{SV}_{H_RA_R} \left( \rho_{H_RA_R}\right),
\end{align}
which defines the output of the real protocol when it does not abort.

Now we aim to simplify $\omega_{\tn{real}}$ to compare it to the output of the ideal protocol. 
Our goal is to show that sharing and verifying the input, performing the distributed circuit and decoding is equivalent to applying the circuit to the inputs directly,
\begin{align}
	\tilde{\mathcal{D}}_{H_RA_R} \circ \dbar{\mathfrak{R}}_{H_RA_R}\circ \mathcal{SV}_{H_RA_R} \left( \rho_{H_RA_R} \right)=\mathfrak{R}_{H_RA_R} (\rho_{H_RA_R}).
\end{align}
Indeed, this follows from Lemma \ref{lem:conserv_enc} and Property \ref{prop:prop1}. By security of the VQSS \cite{Crepeau2002,Smith2001,Lipinska2019}, if the protocol does not abort, there exists a unique double-encoded state after the verification phase, i.e. $\mathcal{SV}_{H_RA_R} \left( \rho_{H_RA_R} \right)$. By definition the decoding  $\tilde{\mathcal{D}}_{H_RA_R}$ is exactly the one performed in Lemma \ref{lem:conserv_enc}. 
Therefore, we have that
\begin{align}
\omega_{\tn{real}}& = (\mathcal{D}_{H_R} \otimes \mathcal{W}_{A_R})\circ \mathcal{E}_{H_RA_R} \circ \tilde{\mathcal{D}}_{H_RA_R} \circ \dbar{\mathfrak{R}}_{H_RA_R}\circ \mathcal{E}_{H_RA_R} \left( \rho_{H_RA_R}\right)\\ 
&=(\mathcal{D}_{H_R} \otimes \mathcal{W}_{A_R})\circ \mathcal{E}_{H_RA_R} \circ \mathfrak{R}_{H_RA_R}\left( \rho_{H_RA_R}\right).
\end{align}
This, together with Equation \eqref{eq:omega_ideal}, gives us that the outputs of the ideal and real protocol are equal for $|B| \leq t$,
\begin{align}
\omega_{\tn{ideal}} = \omega_{\tn{real}}.
\end{align}
Similarly, when $|B|>t$, one can compare  \eqref{eq:ideal_step2} with \eqref{eq:ideal_step5} and obtain that the states are the same for the real and ideal protocol.
\vspace{1em}
What we described so far, considers that the encoding in the sharing phase was performed correctly in the real protocol. However, this does not have to be the case. Every verification performed during the MPQC
has a probability of error inherited from the VQSS. Recall that from Lemma \ref{lem:vqss_sound} the probability of unsuccessful verification in VQSS is lower-bounded by $2^{-\Omega(s)}$. In MPQC we verify:
\begin{itemize}
\item each of the $n$ inputs, 
\item each $\ket{\dbar 0}$ and $\ket{\dbar m}$ necessary to perform the $T$ gate,
\item each $\ket{\dbar 0}$ necessary for the circuit $\mathfrak{R}$
\end{itemize}
Let $\kappa = n + \# T$ gates $+ \#$ ancillas for $\mathfrak{R}$. Then the total probability of error in MPQC is $\kappa 2^{-\Omega(s)}$.

\end{proof}

\end{appendix}

\newpage

\begin{titlepage}
\centering
{\large\bfseries Erratum: Secure multi-party quantum computation with few qubits}

\vspace{1em}

\centering
{\large Victoria Lipinska,$^{1,2}$ Jeremy Ribeiro,$^{1,2}$ and Stephanie Wehner$^{1,2}$}\\
\vspace{2mm}
{\it $^1$QuTech, Delft University of Technology, Lorentzweg 1, 2628 CJ Delft, The Netherlands}\break
{\it $^2$Kavli Institute of Nanoscience, Delft University of Technology, Lorentzweg 1, 2628 CJ Delft, The Netherlands}

\vspace{2em}
\end{titlepage}

We, authors of the article \cite{Lipinska2020}, wish to present this erratum to correct an error we made in a subprotocol verifying magic states. Specifically, the circuit used to implement the verification cannot be implemented transversally for the chosen family of codes. In this erratum we explain the error in detail and \steph{propose a solution based on already existing techniques}, i.e. distributed magic state distillation. This approach increases the operational qubit workspace per node  from $n^2 + 4n$ to $n^2+\Theta(s) n$, where $s$ is the security parameter. The increase is linear in $n$, which means that the main result of our paper remains intact: number of qubits per node necessary to implement the multiparty quantum computation \change{is still smaller than} the previous protocol \cite{Crepeau2005}. \steph{Moreover, the security proof in our manuscript does not change.}

\setcounter{section}{0}
\section{The issue}
In the manuscript \cite{Lipinska2020} we presented a protocol for distributed multiparty quantum computation (MPQC) of universal circuits. Our method is based on quantum error correction and we choose a subfamily of Calderbank-Steane-Shor (CSS) $\ctr$ which allows for a transversal implementation of the Clifford gates. \change{Tranversality is essential since we require that any operation $\Lambda$ implemented by the nodes locally should yield the same operation $\bar{\Lambda}$ at a logical level.}

To implement any circuit we supplement the Clifford gates with logical magic states $\ket{\dbar{m}}$, which need to be verified, i.e. the nodes must collectively agree that there are at most $t$ errors in the logical magic state. After this step the nodes can use the magic state to implement the $T$ gate \change{transversally} using gate teleportation. 

The method we chose for the verification procedure is based on the so called stabilizer measurement. Unfortunately, the controlled gate $C$-$XP^\dagger$ implementing the stabilizer measurement is \emph{not} a Clifford gate, and for our chosen family of codes $\ctr$ this gate is not transversal. This means that the procedure cannot be implemented transversally \change{and the magic state cannot be verified with this procedure.}

\section{The solution}
To solve the problem we propose a new protocol for the verification of magic states. Our new protocol is based on magic state distillation \cite{Bravyi2005} and statistical testing of randomly selected states. We note that this is a method inspired by distillation of entangled pairs whose exact initial state is unknown \cite{Pirker2017}. A similar approach has recently been reported by \cite{Dulek2020}. In the following we first describe the new protocol on a high level and then provide a detailed description. 

In the magic state verification procedure one node produces $M$ copies of the magic state $\ket{m}$. The nodes share and verify the \change{encoding of the} state using the verifiable secret sharing protocol, see Protocol 1 in the manuscript \cite{Lipinska2020}. Then, using public randomness, the nodes pick a fraction of states \change{to be statistically tested. For each picked state, the nodes select a random node who reconstructs the shared state and measures it in the $\{\ket{m},\ket{m^\perp}\}$ basis.}
If all of the measurement results yield $\ket{m}$ then the nodes have statistical evidence that the rest of the states are good with high probability. Next, the nodes perform dephasing in the $\{\ket{m},\ket{m^\perp}\}$ basis on the remaining states. \change{This is done by randomly applying $PX$ gate, which is a Clifford gate, and therefore can be implemented transversally with the code $\ctr$.} This step is necessary, since to perform magic state distillation the initial states must have a diagonal form in the  $\{\ket{m},\ket{m^\perp}\}$ basis. After this, the nodes perform magic state distillation using the 15-to-1 Bravyi-Kitaev protocol \cite{Bravyi2005}. \change{Importantly, this protocol can also be realized using only Clifford gates and measurements, both of which can  be implemented essentially transversally for the chosen family of codes. Note that the distillation procedure can be performed many times to get arbitrarily close to the perfect magic state.} 

We remark that the statistical testing is necessary to assess the quality of magic states before distillation, since the states can be produced by cheaters. If this was the case, performing the distillation straightforwardly would not give a guarantee on the quality of the final magic state.

\section{Implications for the results}

Here we point out the implications of the change we make to the overall results of our work. 

\setcounter{outline}{0}
\begin{framed}
\begin{outline}[\steph{Multi-party quantum computation, repeated from manuscript \cite{Lipinska2020}}]~\\
Input: single-qubit state $\rho_i$ from each node, CSS code $\ctr$ with transversal Cliffords, circuit $\mathfrak{R}$.
\begin{enumerate}
\item \textit{Sharing and verification}\\
Each node $i=1,\dots,n$ encodes her input $\rho_i$ using code $\ctr$ into an $n$-qubit logical state, and sends one qubit (i.e. one single-qubit share) of the logical state to every other node, while keeping one for herself.
The nodes jointly verify the encoding done by node $i$ using verifiable quantum secret sharing protocol (see Protocol 1, manuscript \cite{Lipinska2020}).
\item \textit{Computation}
\begin{itemize}[topsep=0pt,itemsep=-1ex,partopsep=1ex,parsep=1ex]
\item For every Clifford gate in circuit $\mathfrak{R}$:\\
The nodes apply transversal Clifford gates locally to qubits specified by the circuit $\mathfrak{R}$.\\
\item For every $T$ gate in circuit $\mathfrak{R}$ applied to qubit $i$ the nodes run Verification of Magic States, Protocol \ref{prot:verification} (\textbf{the new protocol}). If the verification is successful, the nodes perform Distributed Gate Teleportation, see Protocol 2 in manuscript \cite{Lipinska2020}. 
\end{itemize}
Every $\ket{0}$ ancilla state required for circuit $\mathfrak{R}$, which is prepared by node $i$, is jointly verified by the nodes using verifiable quantum secret sharing, Protocol Protocol 1, manuscript \cite{Lipinska2020}.

If the verification of any step fails the nodes substitute their shares for $\ket{0}$ and abort the protocol at the end of the computation.

\item \textit{Reconstruction}\\
Each node $i$ collects all shares of her part of the output. She corrects errors using code $\ctr$ and reconstructs her output.

\end{enumerate}

\end{outline}
\end{framed}

Let $M$ be the number of magic states produced for each $T$ gate of the circuit one wishes to execute throughout the MPQC. Let $k$ be the number of these states that is measured in Testing of Protocol \ref{prot:verification}, see below. 

\begin{itemize}
    \item  \textit{Security of MPQC.} The security of our MPQC protocol remains unchanged. Theorem \ref{Thm:security} below ensures that conditioned on not aborting and for $M=\Theta(s)$ and $k=\Theta(s)$, the state after verification is $2^{-\Omega(s)}$ close to a logical magic state, where $s$ is the security parameter. This means that, as before, the overall security of the MPQC protocol can be quantified with a total error probability of $\kappa 2^{-\Omega(s)}$, where $\kappa = n + \#T$ gates $+\#$ancillas in a circuit executed in MPQC, see Theroem 1 in the manuscript \cite{Lipinska2020}.
    
    \item \textit{Qubit workspace.} The new scaling for the qubit workspace only has a linear overhead as compared to the previous version, and it remains lower than the previous result by \cite{Crepeau2005}. The qubit workspace required per node is now $n^2 + (M+2)n=n^2+\Theta(s) n$ as opposed to the previously derived $n^2 + 4n$. As before, sharing and verifying the $n$ input qubits uses at most $n^2+2n$ qubits. However, we also must consider the resources needed for the magic state verification. When sharing the magic state, the nodes already hold $n^2$ qubits corresponding to the inputs, and the magic state distribution and verification requires $(M+2)n$ qubits.
    
    \item \textit{Quantum communication complexity.} The quantum communication complexity is essentially unchanged. With the new verification method each node sends  $((M-1)s^2+ k )n \steph{\cdot \#T} = \Theta(s) n s^2 \cdot \#T$ extra qubits. This means that the communication complexity per node is now $\mathcal{O}((n + \# \tn{ancillas} + M \cdot \# T)ns^2 +k n \cdot \#T) =\mathcal{O}((n + \# \tn{ancillas} + \Theta(s) \cdot \# T)ns^2) $, as opposed to the previous $\mathcal{O}((n + \# \tn{ancillas} + \# T)ns^2)$, where $\# T$ is the number of $T$ gates. 
\end{itemize}
\steph{In our MPQC protocol from manuscript \cite{Lipinska2020} we introduced an “abort”
event. That is, the protocol
could abort if there were more than $t$ errors introduced
by the cheaters, accumulated over all inputs.} \change{For comparison, now our MPQC protocol can abort either because the verification of the magic state aborts or because more than $t$ errors have been introduced by the cheaters, accumulated over all inputs. We remark that the verification procedure below can be repeated until successful in the following way. Every time the verification fails two nodes are removed from the next execution: one node who created the state and one of the nodes who measured $\ket{m^\perp}$. By repeating this $t$ times we would remove at most $2t$ nodes, and with certainty remove all of the cheaters. In the next $(t+1)$-th execution, all the $n-2t$ remaining nodes would be honest and the procedure would necessarily succeed. This would leave the MPQC aborting only in the latter case, i.e. when more than $t$ errors occur.}

\onecolumngrid
\vspace{2em}

\setcounter{protocol}{2}
\begin{framed}
\begin{protocol}[Verification of Magic States (VMagic), new protocol] \label{prot:verification}
~\\
\noindent Input: \change{set of apparent cheaters $B$, number $M$ of magic states to be created, number $k$ of magic states to be measured.}\\ 
\noindent Output: verified logical state $\ket{\dbar{m}}$ 

\vspace{1em}
\textbf{Testing}

\begin{enumerate}
\item A randomly selected node $i$ creates $M$ copies of the magic state $\ket{m}$. 
\item The nodes run verifiable secret sharing protocol \change{using code $\ctr$} (VQSS, Protocol 1 in the manuscript \cite{Lipinska2020}) $M$ times, every time with $\ket{m}$ as an input and with dealer $i$. They update the set $B$ with apparent cheaters $B_{m}$ revealed in verifying each copy of $\ket{m}$.  
\item The nodes use public randomness to decide 
\begin{itemize}
    \item which $k$ of the $M$ copies will be measured;
    \item which node will measure each of the selected $k$ copies. 
\end{itemize} 
\item The nodes send the shares according to the division in the previous step, and use the reconstruction of the VQSS \cite{Crepeau2005} to reconstruct a state. 
\item Each node measures the reconstructed state in the $\{\ket{m},\ket{m^\perp}\}$ basis. They announce the results of the measurement. 
\begin{itemize}
    \item If all measurements yield $\ket{m}$, continue.
    \item If any measurement yields $\ket{m^\perp}$, set $B = [n]$ (this will cause the MPQC protocol to abort after the computation phase). 
\end{itemize}
\end{enumerate}

\textbf{Distillation} \change{\cite[Circuit 2.8]{Dulek2020}}

\begin{enumerate}
    \item The nodes use public randomness to apply $PX$ to each share of the remaining $M-k$ logical states \change{with probability $\tfrac{1}{2}$} (n.b. this brings the logical states into a form diagonal in the $\{\ket{m},\ket{m^\perp}\}$ basis).
    \item The nodes use public randomness to permute the remaining $M-k$ logical states.
    \item The nodes apply the 15-to-1 magic state distillation protocol \cite{Bravyi2005}. Any measurements throughout the protocol are broadcasted and the logical value is reconstructed using verifiable classical secret sharing (like in the verification phase of the VQSS, see Protocol 1 in the manuscript \cite{Lipinska2020}). 
\end{enumerate}
\end{protocol}
\end{framed}
\vspace{2em}
\normalsize
\twocolumngrid

\section{Parameter analysis}

In this section we give technical details of the magic state verification protocol, Protocol \ref{prot:verification}. We point out that the security proof in our manuscript does not change. The only alteration is in the \change{derivation of the} error that can be introduced by the verification of the magic state (this corresponds to a different error in the ``real protocol'' at the end of Appendix A in the manuscript \cite{Lipinska2020}). In the following we derive this error explicitly. 
We will first state a few useful lemmas necessary to prove the security of our new verification protocol. Then we will proceed with stating the desired security in Theorem \ref{Thm:security}. 

Let $M$ be the number of magic states distributed by a randomly selected node; let $k < M$ be the number of copies chosen to be measured by all nodes, and $k'\leq k$ out of all of the measured copies be measured by the honest nodes. \steph{Note that selecting a random node and selecting which copies to measure can be done using already assumed public source of randomness and classical multiparty computation, see manuscript \cite{Lipinska2020}.}
The lemma below states that if a state is close to the subspace of state with a small "Hamming weight", then Protocol \ref{prot:verification} can distill it to a state close to a pure magic state. \steph{Since $\{\ket{m}, \ket{m^{\perp}}\}$ is a basis of a qubit space, it follows that any $M-k$ qubit pure state can be written as a superposition of tensor products of vectors in $\{\ket{m}, \ket{m^{\perp}}\}$.} The Hamming weight then needs to be understood as the maximum number of $\ket{m^{\perp}}$ showing in any term of the superposition  \cite{Bouman2010}.

\setcounter{lemma}{2}
\begin{lemma}[Lemma 2.7 of \cite{Dulek2020}]
Let $ V_\delta  := {\rm  span}\{P_\pi (\ket{T}^{\otimes (M-k)-w} \ket{T^\perp}^{\otimes w}) : \tfrac{w}{M-k} \leq \delta, \pi \in S_{M-k}\}$, where $S_{M-k}$ is the set of permutation of $M-k$ elements, and $P_{\pi}$ is the operator that permutes $M-k$ qubits according to the permutation $\pi$. Let $\Pi_{V_\delta}$ be a projector onto $V_\delta$. Let $\Xi$ be the CPTP map describing the action of Distillation of Protocol \ref{prot:verification}.
Let $\rho$ be a $M-k$ qubit state such that $\Tr(\Pi_{V_\delta} \rho) \geq 1- \epsilon$, then,

\begin{align}
    \big\|\Xi(\rho) - \ketbra{T}{T}\big\|_1 \leq O\Big((M-k)(\sqrt{35} \delta)^{(M-k)^c/2} + \epsilon \Big),
\end{align}
where $c\approx 0.406$ and $\delta$ is chosen such that $\delta \leq 0.14$.
\end{lemma}

Now we will prove a lemma lower-bounding the number  of copies of the magic state that must be measured in Testing, Protocol \ref{prot:verification}, such that at least some minimum number $s'$ of them is measured by honest nodes with high probability.

\begin{lemma}\label{Lmm:How_big_k}
Let $\mu \in (0,1)$, let $s'\geq 1$ be some integer. Let $H:=(1-\tfrac{1}{n}\lfloor \tfrac{n-1}{4} \rfloor) \in [3/4, 1]$. In the testing procedure Protocol \ref{prot:verification}, \steph{if the number of measured copies of the magic state} $k$ satisfies the following,
\begin{align}\label{eq:size_k}
\begin{split}
    k \geq \frac{2 H s' + \ln(\mu^{-1})/2+\sqrt{2 H s' \ln(\mu^{-1})+ \ln^2(\mu^{-1})/4}}{2H^2}
\end{split}    
\end{align}
then,
\begin{align}
    \Pr(k'<s') \leq \mu,
\end{align}
\steph{where $k'$ is the number of copies of the magic state measured by the honest nodes.} 
\end{lemma}
\begin{proof}
 For a sequence of IID Bernoulli random variables $X_1,\ldots, X_n$, and some $\lambda \in [0,1]$ Hoeffding's inequality \cite{Hoeffding1963} insures that,
 \begin{align}
     \Pr \left( \sum_1^k X_i \leq \mathbb{E}\left( \sum_1^k X_i\right) - \lambda k\right) \leq e^{-2 \lambda ^2 k}.
 \end{align}
 In words, the Hoeffding's inequality bounds the probability that the fraction of observed $1$s in the sequence of random variables deviates from its expectation value by more than $\lambda k$. In Protocol \ref{prot:verification}, we can define a Bernoulli variable for each measured copy of the magic state as follows: The random variable $X_i$ takes value $1$ if and only if the copy $i$ is send to an honest node. Therefore, we have $k' = \tfrac{1}{k} \sum_1^k X_i$ and $\mathbb{E}(\tfrac{1}{k} \sum_1^k X_i)=(1-\tfrac{1}{n}\lfloor \tfrac{n-1}{4} \rfloor) = H$. Plugging this in to Hoeffding's inequality we get,
 \begin{align}
     \Pr \left( k' \leq (H - {\lambda}) k\right) \leq e^{-2 {\lambda} ^2 k}.
 \end{align}
 Then by choosing ${\lambda}$ and $k$ such that ${\lambda} = \tfrac{\mathbb{E}(\sum_1^k X_i)-s'}{k}= \tfrac{H k -s'}{k}$ and $e^{-2{\lambda}^2 k} = \mu$ we get,
 \begin{align}
     \Pr(k'<s') \leq \mu,
 \end{align}
 and that $k$ must satisfy,
 \begin{align}
     H^2 k^2 - (2H s' +\ln(\mu^{-1})/2) k +s'^2 \geq 0,
 \end{align}
 from which, by solving the inequality for $k$, we get inequality \eqref{eq:size_k}. \steph{Note that since $\mu$ can take any value in $(0,1)$, the probability $\Pr(k'<s')$ can be made arbitrarily small.}
\end{proof}

Finally we restate a Theorem from \cite{Bouman2010} saying that if Testing of Protocol \ref{prot:verification} does not abort then the state \textit{before} Testing was already close to a space of states with a small Hamming weight.

\begin{lemma}[From Theorem 3 of \cite{Bouman2010}] \label{Lmm:Serge-Ferh}
    Let $\ket{\phi_{AE}} \in (\mathbb{C}^2)^{M}\otimes \mathcal{H}_E$ be  a  quantum state  and  let $\beta=\{\ket{v_0}, \ket{v_1}\}$ be  a  fixed  single-qubit  basis.   If  we  measure $k$ random  qubits  of $\Tr_E(\ketbra{\phi_{AE}})$ in  the $\beta$-basis  and  all  of  the  outcomes  are $\ket{v_0}$,  then  with  probability $1-e^{-\delta^2k}$, we have that
    \begin{align}
    \begin{split}
        \ket{\phi_{AE}} \in {\rm  span}\{& P_\pi (\ket{T}^{\otimes M-w} \ket{T^\perp}^{\otimes w})\otimes \ket{\psi} : \\
        &\tfrac{w}{M} \leq \delta, \pi \in S_{M}, \ket{\psi} \in \mathcal{H}_E\}
    \end{split}
    \end{align}
\end{lemma}

The lemma above has a useful corollary, namely that the state of remaining unmeasured qubits \textit{after} Testing in Protocol \ref{prot:verification} is close to a space of states with a small Hamming weight.

\begin{corollary}[From Lemma \ref{Lmm:Serge-Ferh}]
    Let $\ket{\phi_{AE}} \in (\mathbb{C}^2)^{M}\otimes \mathcal{H}_E$ be  a  quantum state  and  let $\beta=\{\ket{v_0}, \ket{v_1}\}$ be  a  fixed  single-qubit  basis. If among $k$ randomly chosen  qubits  of $\Tr_E(\ketbra{\phi_{AE}})$,  $k'$ of them are correctly measured in the $B$-basis and all of  the  outcomes are $\ket{v_0}$ while $k-k'$ are measured with an arbitrary POVM, then the state $\rho$ on the remaining $(M-k)$ qubits of $\Tr_E(\ketbra{\phi_{AE}})$ is $e^{-\delta^2 k'}$-close to the subspace 
\begin{align}
\begin{split}
    V_{\delta} = {\rm span}\{& P_\pi (\ket{T}^{\otimes (M-k)-w} \ket{T^\perp}^{\otimes w}): \\
    & \tfrac{w}{M-k} \leq \delta, \pi \in S_{M-k}\}.
\end{split}
\end{align}
    In other words, if $\Pi_{V_\delta}$ is a projector on the above subspace, then $\Tr(\Pi_{V_\delta} \rho)\geq 1 - e^{-\delta^2 k'}$.
\end{corollary}

Now we are ready to state the security of our new verification procedure. 
\begin{thm}\label{Thm:security}
 Let $\Gamma$ be the CPTP map describing the action of Testing, Protocol \ref{prot:verification} and $\Xi$ be a CPTP map describing the action of Distillation, Protocol \ref{prot:verification}. Then we have,
 
  \begin{align}\label{eq:final_sec}
  \begin{split}
      \big\| \Xi \circ \Gamma(\rho &)_{|\textnormal{not abort}}  - \ketbra{T} \big\|_1 \\ 
      &\leq O\Big({(M-k)(\sqrt{35} \delta)^{(M-k)^c/2}} + e^{-\delta^2 s'} + \mu \Big),     
  \end{split}
 \end{align}
 where $c\approx 0.406$. Recall that $s$ is the security parameter of the whole MPQC protocol. By setting $M-k = s$, $s'=s$, and $\mu=2^{-s}$ Equation \eqref{eq:final_sec} becomes,
 \begin{align}\label{eq:sec_s}
 \begin{split}
    \big\| \Xi \circ \Gamma(\rho)_{|\textnormal{not abort}} - \ketbra{T} \big\|_1 \leq&  O(2^{-\Omega(s)} + e^{-\delta s^2 } + 2^{-s} )\\
     =& 2^{-\Omega(s)}.
 \end{split}
 \end{align}
 \end{thm}
 Note that, by Lemma \ref{Lmm:How_big_k}, setting $s'=s$ and $\mu=2^{-s}$ forces $k$ to satisfy,
 \begin{align}
 \begin{split}
     k \geq& \frac{2 H s + s/2 + \sqrt{2Hs^2 +s^2/4}}{2H^2} \\
     =& \left(\frac{2 H  + 1/2 + \sqrt{2H +1/4}}{2H^2} \right)  s =  \Theta(s).     
 \end{split}
 \end{align}
Overall, we have that Equation \eqref{eq:sec_s} holds for $M= (M-k) + k = \Theta(s)$.

\section{Acknowledgements}
We thank J.~G.~H\"olting for point us to the error in the original manuscript. We thank J.~Helsen and B.~Dirkse for useful comments and feedback on this erratum.


%

\end{document}